\documentclass[11pt]{article}
\usepackage{fullpage}
\usepackage{amssymb}
\usepackage{cancel}
\IfFileExists{ulem.sty}{\usepackage[normalem]{ulem}}{}
\usepackage{braket}
\usepackage{ragged2e}
\usepackage{amsmath}
\IfFileExists{algorithm.sty}{\usepackage{algorithm}}{}
\IfFileExists{algpseudocode.sty}{\usepackage[noend]{algpseudocode}}{}
\usepackage{amsmath, amsthm}
\usepackage{hyperref}
\usepackage{cleveref}
\usepackage{color}
\IfFileExists{soul.sty}{\usepackage{soul}}{}

\usepackage{makeidx}
\usepackage{amsfonts,latexsym,xspace,graphicx,float,mathtools,epigraph}
\usepackage{enumitem}
\usepackage{tikz}
\usepackage{booktabs}
\usepackage{tabularx}
\usepackage{array}
\usepackage{amsmath}

\IfFileExists{bbm.sty}{\usepackage{bbm}}{\let\mathbbm\mathbb}

\newcolumntype{Y}{>{\raggedright\arraybackslash}X}
\usetikzlibrary{positioning, arrows.meta}
\makeindex

\newif\ifnotes\notesfalse

\newtheorem{theorem}{Theorem}
\newtheorem{proposition}{Proposition}
\newtheorem{lemma}{Lemma}
\newtheorem{claim}{Claim}
\newtheorem{remark}{Remark}

\newtheorem{corollary}{Corollary}
\newtheorem{definition}{Definition}

\newtheorem*{informaltheorem}{Theorem} 
\newtheorem*{informalconjecture}{Conjecture}

\newtheorem{conjecture}{Conjecture}

\newtheorem{note}{Note}

\def\norm#1{\mathopen\| #1 \mathclose\|}		

\makeatletter
\def\Bigbar#1{\mathrel{\left|\vphantom{#1}\right.\n@space}}

\makeatother

\renewcommand{\Pr}{\mathop{\bf Pr\/}}

\newcommand{\poly}{\mathrm{poly}}

\newcommand{\dist}{\mathrm{dist}}

\newcommand{\supp}{\mathrm{supp}}

\newcommand{\Z}{\mathbbm Z}
\newcommand{\F}{\mathbbm F}

\newcommand{\NP}{\mathsf{NP}} 

\ifnotes
\usepackage{xcolor}
\definecolor{mygrey}{gray}{0.50}
\newcommand{\notename}[3]{{\textcolor{#3}{\footnotesize{\bf (#1:} {#2}{\bf ) }}}}

\else

\newcommand{\notename}[3]{{}}

\fi

\IfFileExists{stmaryrd.sty}{\usepackage{stmaryrd}}{\newcommand{\llbracket}{[\![}\newcommand{\rrbracket}{]\!]}}
\renewcommand{\th}{\widetilde{h}}
\newcommand{\DP}{\ensuremath{\textrm{DP}}}
\newcommand{\QDP}{\ensuremath{\textrm{QDP}}}

\newcommand{\hf}{\widetilde{f}}
\newtheorem{problem}{Problem}
\newcommand{\C}{\mathcal{C}}
\newcommand{\mat}[1]{\ensuremath{\boldsymbol{#1}}}
\newcommand{\cv}{{\mat{c}}}
\newcommand{\Unif}{\leftarrow}
\newcommand{\ev}{\mat{e}}

\newcommand{\Hm}{\ensuremath{\mathbf{H}}}
\newcommand{\uv}{\mat{u}}
\newcommand{\yv}{{\mat{y}}}
\newcommand{\RS}{\ensuremath{\textrm{RS}\,}}

\newcommand{\Iint}[2]{\llbracket #1 , #2 \rrbracket}
\newcommand{\hu}{\widehat{u}}
\newcommand{\Gm}{\ensuremath{\mathbf{G}}}
\renewcommand{\aa}{\mathcal{A}}
\newcommand{\Time}{\mathrm{Time}}

\newcommand{\BDD}{\ensuremath{\mathrm{BDD}}}
\newcommand{\IBDD}{\ensuremath{\mathrm{IBDD}}}
\newcommand{\ketbra}[2]{|#1\rangle\langle#2|}
\newcommand{\altketbra}[1]{\ketbra{#1}{#1}}
\newcommand{\kb}[1]{\altketbra{#1}}
\renewcommand{\uv}{\ensuremath{\mathbf{u}}}
\newcommand{\sv}{\ensuremath{\mathbf{s}}}
\newcommand{\xv}{\ensuremath{\mathbf{x}}}
\renewcommand{\hf}{\widehat{f}}
\newcommand{\zv}{\ensuremath{\mathbf{z}}}

\newcommand{\COMMENT}[1]{}

\newcommand{\SIS}{\ensuremath{\mathrm{SIS}}}
\newcommand{\DLOG}{\ensuremath{\mathrm{DLOG}}}

\begin{document} 
	\title{Regev's reduction as a candidate quantum algorithm for the discrete logarithm problem in finite abelian groups}
	\author{M. Isabel Franco Garrido\thanks{Institute for Quantum Information and Matter, California Institute of Technology. mfrancog@caltech.edu } , \quad
		André Chailloux\thanks{Inria de Paris. andre.chailloux@inria.fr} }

	\date{}
	\maketitle

\begin{abstract}
In this article, we investigate whether combining Regev's reduction with the Cheng--Wan reduction can yield an alternative quantum algorithm for the discrete logarithm problem (DLOG) in finite abelian groups. Cheng and Wan reduce DLOG over finite fields to bounded-distance
decoding of low-rate Reed--Solomon codes at a large decoding radius. Regev's
reduction, in turn, exploits a decoder for the high-rate dual code with a lower noise level to obtain a quantum algorithm for decoding the original low-rate code
at a larger radius. The Cheng--Wan instances therefore provide a natural
setting in which to explore the power of Regev's framework.

We first extend the Cheng--Wan reduction to the discrete logarithm problem in finite abelian groups, assuming access to embeddings of their cyclic factors.
We also establish $\NP$-hardness of Reed--Solomon bounded-distance decoding
at asymptotically vanishing rate, reaching the rate regime relevant to
Cheng--Wan, while the radius at which we establish hardness remains beyond the
Cheng--Wan radius. We then instantiate Regev's reduction at the Cheng--Wan parameters
and, under a conjecture on the performance of the induced BDD algorithm on the structured Cheng--Wan instances, characterize the noise threshold for the quantum decoding problem (QDP) sufficient to solve DLOG.

We evaluate known decoders in this regime and show that a quantum decoder based on unambiguous state discrimination (USD) enables Regev's reduction to reach a strictly stronger decoding regime than the standard classical Reed--Solomon decoders we consider. Nevertheless, the QDP noise threshold required by the Cheng--Wan reduction is asymptotically \(4/3\) times that achieved by USD. Going beyond the USD framework, we show that every Cheng--Wan instance can be solved through Regev's reduction using a suitable Pretty Good Measurement and its inverse, although no efficient general implementation of the required measurement is known.
\end{abstract}

	\newpage
	\tableofcontents
	\newpage

	\section{Introduction}
	\subsection{Regev's reduction and its code-based incarnation}
Regev's reduction~\cite{Reg05}, originally formulated for lattices and later adapted to codes~\cite{DRT23}, can be used to solve bounded-distance decoding (BDD) on a classical linear code \(\C\subseteq\mathbb F_q^n\). Given a received word \(\yv_0\) promised to lie within Hamming distance \(t\) of \(\C\), the goal is to recover a codeword \(\cv\in\C\) within distance \(t\) of \(\yv_0\). At a high level, the reduction coherently samples from the coset \(\yv_0-\C\), with the choice of a filter function \(f\) controlling the resulting distribution of the sampled elements. Each element of this coset has the form \(\yv_0-\cv\) for some \(\cv\in\C\), so sampling one of Hamming weight at most \(t\) identifies a codeword within distance \(t\) of \(\yv_0\). The reduction implements this sampler by combining the quantum Fourier transform with a solver for \(\QDP(\C^\perp,f)\). Schematically, the reduction takes the following form:

\begin{figure}[H] 
\centering \begin{tikzpicture} \node[ draw, rectangle, minimum width=3.8cm, minimum height=1.4cm, align=center ] (left) {An efficient solver for\\ \(\QDP(\C^\perp,f)\)}; \node[ draw, rectangle, minimum width=3.8cm, minimum height=1.4cm, align=center, right=2.5cm of left ] (right) {A quantum algorithm for\\ \(\BDD(\C,t)\)}; \draw[ double, double equal sign distance, -Implies, thick ] (left) -- (right); 
\end{tikzpicture} 
\caption{An efficient solver for \(\QDP(\C^\perp,f)\) yields, through Regev's reduction, a quantum algorithm for \(\BDD(\C,t)\).} \label{fig:regev-reduction} \end{figure}


The reduction constructs an instance of \(\QDP(\C^\perp,f)\) in which, for an unknown codeword \(\cv^\perp\in\C^\perp\), the input is \[ \ket{\Psi_{\cv^\perp}}=\sum_{\ev\in\mathbb{F}_q^n}f(\ev)\ket{\cv^\perp+\ev}, \] and the goal is to output \(\cv^\perp\); in our application, the output must be written to an auxiliary register. A classical decoder can be implemented reversibly and applied in superposition, writing its output to the auxiliary register. A quantum decoder may instead exploit the coherent structure of the input state and potentially outperform reversible implementations of classical decoders. Prior works~\cite{CT24,chailloux2024quantumadvantagesoftdecoders,chailloux2025opixsoftdecoders} showed that, in certain regimes, processing the QDP input coherently can achieve a higher decoding success probability than applying a reversible implementation of a classical decoder in superposition.


Whether this improvement translates into an overall quantum advantage through Regev’s reduction, however, depends on the classical complexity of the resulting target problem.

This has been a central obstacle in previous applications of the reduction: the classical hardness of the resulting problems generally remains unproved. Chen, Liu, and Zhandry~\cite{CLZ22} gave a quantum algorithm for \(\SIS_\infty\) at parameters for which no efficient classical algorithm was known, but classical algorithms were subsequently found for these parameters~\cite{KOW25}. Subsequently,~\cite{YZ24} established a provable quantum advantage, although only for an oracle problem.
	
More recently, Jordan et al.~\cite{jordan2025optimizationdecodedquantuminterferometry} introduced Decoded Quantum Interferometry and obtained evidence of a heuristic quantum advantage for Optimal Polynomial Interpolation (OPI). This problem is closely related to Reed--Solomon decoding and can be viewed as a structured Max-LinSAT problem with Vandermonde constraints. However, the classical complexity of the relevant OPI instances remains insufficiently understood. 

Thus, while existing results suggest the possibility of a quantum advantage, they do not establish one for a concrete application with independently supported classical hardness. 
This raises a central question: can Regev's reduction be applied to a natural decoding task arising from a standard computational problem believed to be classically hard?
    
\subsection{The Cheng--Wan reduction}

Motivated by the prospect of establishing a quantum advantage through Regev’s reduction, we investigate its performance on the BDD instances arising from the Cheng--Wan reduction~\cite{cheng2007list}. Cheng and Wan reduce discrete logarithm in \(\F_{q^h}^{\times}\) to \(\BDD\) on low-rate Reed--Solomon codes, so these instances provide a natural setting in which to study Regev's reduction on a structured decoding task whose efficient solution would yield an efficient algorithm for discrete logarithm, a canonical problem for which no classical polynomial-time algorithm is known in the relevant parameter regimes.

Like OPI, Reed--Solomon \(\BDD\) can be viewed as a structured Max-LinSAT problem with Vandermonde constraints. This structural similarity provides an additional reason to investigate whether the heuristic quantum advantage observed for OPI may extend to this setting. Moreover, Reed--Solomon codes provide a particularly natural setting in which to evaluate Regev's reduction: their duals remain in the generalized Reed--Solomon family and, for the full-support codes considered here, are again ordinary Reed--Solomon codes. This allows us to directly relate the performance of QDP decoders on the dual code to the \(\BDD\) radius achieved on the primal code.

Since Shor's algorithm~\cite{shor1997polynomial} already solves discrete logarithm efficiently, our goal is not to obtain a new asymptotic quantum speedup. Rather, we ask whether discrete logarithm can be solved through the structurally different mechanism of Regev's reduction and coherent decoding. Algorithms arising from the DQI/Regev framework may also exhibit resource tradeoffs different from those of Shor's algorithm, which is relevant in light of recent concrete resource estimates~\cite{khattar2025verifiable}.
    
To investigate this alternative route to discrete logarithm, we first extend the Cheng--Wan reduction to a class of finite abelian groups, bringing its setting closer to the general setting in which Shor's algorithm applies~\cite{shor1997polynomial}. We then turn to the associated decoding problem and investigate the hardness landscape of Reed--Solomon \(\BDD\) in the low-rate regime relevant to the reduction. Specifically, we extend known NP-hardness results for constant-rate Reed--Solomon codes~\cite{GandikotaGhaziGrigorescu2018} to the asymptotically zero-rate regime. This result does not establish NP-hardness of the Cheng--Wan instances: the codes appearing in our NP-hardness result are not the full-support Reed--Solomon codes used by Cheng and Wan, and the corresponding decoding radius remains well separated from the Cheng--Wan radius.

Having clarified the surrounding hardness landscape, we next instantiate Regev's reduction on the Cheng--Wan instances and evaluate currently known decoders for the resulting QDP task. For the decoders considered, we determine a BDD radius that the induced algorithm achieves with non-negligible probability on uniformly random received words. We are not aware of a mechanism by which the structure of the Cheng--Wan instances would cause this guarantee to fail. Nevertheless, we cannot prove that the guarantee extends to the structured Cheng--Wan distribution and therefore formulate this transfer as a conjecture.

We next evaluate efficient decoders for the QDP task arising from our instantiation of Regev's reduction. To obtain the desired Hamming-weight distribution in the coset sampler, we choose \(f\) to be an i.i.d.\ Bernoulli filter function whose amplitudes depend only on the Hamming weight. We then analyze the best currently known decoders for the resulting QDP task on the high-rate dual Reed--Solomon code. All fall short of the Cheng--Wan threshold by a constant factor, although unambiguous state discrimination improves on the classical decoders considered here. We also characterize the QDP performance sufficient for the induced BDD algorithm to reach this threshold and, under the transfer conjecture, yield an algorithm for discrete logarithm.

Finally, moving beyond i.i.d.\ Bernoulli filter functions, we show that every Cheng--Wan instance can be solved through Regev's reduction using a suitable Pretty Good Measurement, without relying on the conjecture. The same conclusion applies to the promised YES instances in our NP-hard parameter regime. This demonstrates that the general framework is expressive enough to solve every promised instance in principle, although it does not provide an efficient implementation of the required measurement.
    
	\subsection{Contributions}
	
	Our first contribution is to extend the Cheng--Wan reduction to finite abelian groups given as explicit products of cyclic subgroups of finite-field multiplicative groups. The finite-field encodings and primitive generators of the ambient multiplicative groups are supplied as part of the input. Given this decomposition, we compute the discrete logarithm independently in each cyclic factor, and combining the resulting exponent components yields the discrete logarithm in the full group. We formalise this extension in the following theorem. 
    
    In what follows, $\RS[n,k]_q$ denotes the Reed--Solomon code of length $n$ and dimension $k$ over $\F_q$, $\BDD(\C,t)$ asks for a codeword of $\C$ within Hamming distance $t$ of a received word, and $\DLOG_C$ denotes the discrete logarithm problem in a cyclic group $C$. Formal definitions are given in Sections~\ref{sec:preliminaries} and~\ref{sec:DLOG_abelian}.
	
	\begin{informaltheorem}[Theorem~\ref{thm:abelian-DLOG-BDD}, informal]
        Let $G$ be a finite abelian group with a cyclic factor $C\le\F_{q^h}^\times$, where $2\le h\le q^{1/4}-2$. Assume an explicit finite-field encoding and a supplied primitive generator of $\F_{q^h}^\times$. A decoder for $\BDD(\RS[q,3h+4]_q,q-4h-4)$ that runs in time polynomial in $q$ and succeeds with inverse-polynomial probability on the Cheng--Wan distribution yields an algorithm for $\DLOG_C$ in expected time polynomial in $q$. Applying this to every cyclic factor solves $\DLOG_G$, with the component costs added.
    \end{informaltheorem}
	
	We then ask where the worst-case hardness of Reed--Solomon $\BDD$ begins. Gandikota, Ghazi and Grigorescu~\cite{GandikotaGhaziGrigorescu2018} proved $\NP$-hardness at constant rate; we extend their finite-field construction to asymptotically zero rate by padding the associated Moments Subset Sum instance. The resulting family differs from the full-support Cheng--Wan family. Its $\NP$-hard decoding radius is substantially larger than the Cheng--Wan radius, which is not known to be $\NP$-hard.
	
	\begin{informaltheorem}[Theorem~\ref{thm:ggg18}, informal]
        For a family of Reed--Solomon codes over explicitly represented finite fields, decision $\BDD(\RS[n,k]_q,n-k-d)$ is $\NP$-complete for $1\le d\le c\log n/\log\log n$, for a sufficiently small constant $c>0$, even with $k=\Theta(n^{1/4})$.
    \end{informaltheorem}
	
	We then apply Regev's reduction to Reed--Solomon codes in the Hamming metric, following~\cite{chailloux2025opixsoftdecoders}. This reduction targets inhomogeneous $\BDD$, where the received word is uniformly random. Its success depends on the Fourier mass within the target decoding radius, as stated in the following theorem.
	
	\begin{informaltheorem}[Theorem~\ref{thm:RS-instantiation}, informal]
        Suppose an efficient quantum algorithm solves $\QDP(\RS[q,q-k]_q,\tau)$ with probability $P_{\mathrm{Dec}}\ge1/\poly(q)$. For an integer $1\le g\le q$, let
        \[
        \tau^\perp=\frac1q\left(\sqrt{(q-1)(1-\tau)}-\sqrt\tau\right)^2,
        \qquad \eta=\Pr[Z<g],\quad Z\sim\operatorname{Bin}(q,1-\tau^\perp).
        \]
        There is an efficient quantum algorithm for $\IBDD(\RS[q,k]_q,q-g)$ with success probability at least
        \[
        P_{\mathrm{Dec}}(1-\eta)-2\sqrt{\eta P_{\mathrm{Dec}}(1-P_{\mathrm{Dec}})}.
        \]
        In particular, the success probability is $(1-o(1))P_{\mathrm{Dec}}$ when $\eta=o(P_{\mathrm{Dec}})$.
    \end{informaltheorem}

	Next, we characterize the decoding performance required to solve $\DLOG$ through this approach. The general reduction guarantees success for a uniformly random received word $\yv_0$, whereas the Cheng--Wan reduction produces received words from a structured distribution. We therefore make the following assumption on the induced BDD algorithms:

\begin{informalconjecture}[Conjecture~\ref{ass:heuristic}, informal] \label{ass:heuristic_informal} At the Cheng--Wan parameters, the BDD algorithm induced by a QDP decoder through the reduction of \Cref{thm:RS-instantiation} retains inverse-polynomial success probability on the structured Cheng--Wan distribution whenever it succeeds with inverse-polynomial probability on uniformly random received words. \end{informalconjecture} 

This conjecture allows us to state the following theorem:
	
	\begin{informaltheorem}[Theorem~\ref{thm:DLOG-via-Regev-heuristic}, informal]
        Under Conjecture~\ref{ass:heuristic}, an efficient algorithm for $\QDP$ on $\RS[q,q-(3h+4)]_q$ with inverse-polynomial success probability at noise slightly above $4h/q$ yields an algorithm for the discrete logarithm problem in $\F_{q^h}^\times$ in expected time polynomial in $q$, in the parameter regime of Theorem~\ref{thm:DLOG-via-Regev-heuristic}.
    \end{informaltheorem}

	We use the following efficient decoders: Berlekamp--Welch~\cite{WelchBerlekamp1986}, Guruswami--Sudan~\cite{guruswami1998improved}, and unambiguous state discrimination~\cite{Ivanovic1987}. Table~\ref{Table:1I} summarizes the $\BDD$ radius each one reaches; all of them fall short of the threshold required by the Cheng--Wan reduction by a constant factor.
	
	\begin{table}[h!]
        \centering
        \renewcommand{\arraystretch}{1.4}
        \small
        \begin{tabular}{|l|c|c|l|}
            \hline
            Decoder & QDP noise $\tau$ & IBDD error fraction $\tau'$ & Remark \\
            \hline
            Berlekamp--Welch & $\approx\tfrac{3\th}{2}$ & $\approx1-\tfrac{3\th}{2}$ & \\
            Guruswami--Sudan & $\approx\tfrac{3\th}{2}$ & $\approx1-\tfrac{3\th}{2}$ & \\
            USD & $\approx3\th$ & $\approx1-3\th$ & \\
            \hline
            Required (Cheng--Wan) & $\approx4\th$ & $\approx1-4\th$ & DLOG, under Conj.~\ref{ass:heuristic} \\
            \hline
        \end{tabular}
        \caption{QDP noise on $\RS[q,q-k]_q$ and resulting IBDD error fraction on $\RS[q,k]_q$, with $k=3h+4$, $\widetilde h=h/q$ and $h/\log q\to\infty$. The entries give the leading terms; the success guarantee includes the tail margin in~\eqref{eq:margin}. The exact Cheng--Wan radius is $q-4h-4$. The three decoding guarantees fall short of this target by a constant factor.}
        \label{Table:1I}
    \end{table}

	Our final contribution is to show that the Pretty Good Measurement can solve every promised $\BDD$ instance through Regev's reduction.
	
	\begin{informaltheorem}[Theorem~\ref{thm:PGM-solves-BDD}, informal]
        For every promised $\BDD$ instance, there is an efficiently preparable error state such that coherent implementation of its Pretty Good Measurement and its inverse solves the instance through Regev's reduction.
    \end{informaltheorem}
	
	This applies to the Cheng--Wan instances and to the promised YES instances in the $\NP$-hard families of Theorem~\ref{thm:ggg18}. The theorem does not provide an efficient implementation of the PGM.
	
\medskip

	\noindent \textbf{Structure.} The paper is organised as follows. Section~\ref{sec:preliminaries} introduces the necessary background on notation, basic algebra, linear codes, Reed--Solomon codes, and characters and Fourier analysis. Section~\ref{sec:DLOG_abelian} reviews the discrete logarithm problem in finite abelian groups. In particular, we formalise the definition used throughout the paper and survey the known classical algorithms for solving this problem.
	Section~\ref{sec:classical-red} presents the classical Cheng--Wan reduction from discrete logarithm to decoding, surveys the main classical decoding algorithms for Reed--Solomon codes, and describes the complexity-theoretic limits of the problem, including $\NP$-hardness results for low-rate Reed--Solomon codes. Section~\ref{sec:quantum-red} develops the quantum reduction, instantiates it for Reed--Solomon codes with several decoding approaches including USD, and identifies the QDP parameter that would yield an algorithm for the discrete logarithm problem under Conjecture~\ref{ass:heuristic}. Section~\ref{sec:PGM} proves the PGM result for every promised BDD instance, using an instance-dependent error state and without an efficient implementation guarantee.

	\section{Preliminaries}\label{sec:preliminaries}
	
	\subsection{Notation}
	
	Plain letters (e.g.\ $x,y,z$) denote scalars, typically elements of $\F_q$. Bold lowercase letters (e.g.\ $\xv, \yv$) denote vectors, typically in spaces such as $\F_q^n$, $\F_q^h$, etc., and also group elements. Bold capital letters (e.g.\ $\mathbf{A}, \mathbf{B}$) denote matrices.
	
	We write $\F_q^\times$ for the set of all nonzero elements of $\F_q$, i.e.\ the multiplicative group $\F_q \setminus \{0\}$. We write $\F_q[x]$ for the ring of univariate polynomials over $\F_q$ in the indeterminate $x$. For $\alpha$ algebraic over $\F_q$, $\F_q[\alpha]$ denotes the smallest subring containing $\F_q$ and $\alpha$, and $\F_q(\alpha)$ denotes the smallest subfield containing $\F_q$ and $\alpha$.
	
	Throughout, we use the following reserved symbols: $G$ denotes a general group, $C$ a cyclic group, $\mathcal{C}$ a code, and $\mathcal{C}^\perp$ its dual. We write $\operatorname{ord}(g)$ for the order of an element $g$, and, when no confusion can arise, also for the order of the cyclic group $C$. We use $\langle g \rangle$ for the subgroup generated by $g$. We reserve $p$ for a prime and $q$ for a prime power. We write $Z\sim\operatorname{Bin}(n,p)$ for a binomial random variable with $n$ trials and success probability $p$.
	
	\subsection{Finite fields and extensions}
	
	\begin{definition}[Irreducible polynomial]
		Let $F$ be a field and let $f(x) \in F[x]$ be a non-constant polynomial. We say that $f(x)$ is irreducible over $F$ if whenever $f(x) = g(x)h(x)$ with $g(x), h(x) \in F[x]$, then either $g(x)$ or $h(x)$ is a unit in $F[x]$ (i.e.\ a nonzero constant polynomial). Equivalently, $f(x)$ has no factorisation into polynomials of strictly smaller positive degree in $F[x]$.
	\end{definition}
	
	\begin{definition}[Quotient field]
		Let $\mu(x) \in \F_q[x]$ be irreducible of degree $h$. The quotient ring $\F_q[x]/(\mu)$ is a field of size $q^h$.
	\end{definition}
	
	\begin{definition}[Canonical root]
		Let $\pi : \F_q[x] \to \F_q[x]/(\mu)$ be the quotient map. We denote by $\alpha$ the residue class of $x$, i.e.\ $\alpha := \pi(x)$. Then $\mu(\alpha) = 0$ in $\F_q[x]/(\mu)$.
	\end{definition}
	
	\begin{lemma}[Power basis of a finite extension]
		Every element of $\F_q[x]/(\mu)$ can be written uniquely as $a_0 + a_1 \alpha + \cdots + a_{h-1} \alpha^{h-1}$ with $a_i \in \F_q$. In particular, $\{1, \alpha, \ldots, \alpha^{h-1}\}$ is an $\F_q$-basis and $\F_q[x]/(\mu) = \F_q[\alpha] = \F_q(\alpha)$.
	\end{lemma}
	\begin{proof}
		This is the standard description of an algebraic extension; see e.g.\ \cite[Ch.~2]{LidlNiederreiter1997}.
	\end{proof}
	
	We use the standard correspondence between finite extensions of $\F_q$ and irreducible polynomials over $\F_q$: every field $K$ with $|K| = q^h$ is isomorphic to $\F_q[x]/(\mu)$ for some irreducible $\mu$ of degree $h$, obtained as the minimal polynomial of a generator $\beta$ of $K$ over $\F_q$. This identification depends on the choice of $\mu$ (equivalently, the choice of generator).
	
	In what follows, we fix one such irreducible polynomial $\mu$ and identify $\F_{q^h}$ with $\F_q[x]/(\mu)$. Whenever we write $\mathbf{b} \in \F_{q^h}^\times$, we implicitly also refer to its representative $b(\alpha) \in \F_q[\alpha]$.
	
	\subsection{Linear codes}
	
	A $q$-ary linear code $\C$ of dimension $k$ and length $n$ is characterised by a full-rank generator matrix $\Gm \in \F_q^{k \times n}$ or, equivalently, by a full-rank parity-check matrix $\Hm \in \F_q^{(n-k) \times n}$, and we write
	\[
	\C = \{\xv \Gm : \xv \in \F_q^k\}
	\quad\text{or}\quad
	\C = \{\yv \in \F_q^n : \Hm \yv^\top = \mathbf{0}\}.
	\]
	Each code $\C$ has a \emph{dual code} $\C^\perp = \{\yv \in \F_q^n : \forall \cv \in \C,\; \yv \cdot \cv = 0\}$; any generator matrix of $\C$ is a parity-check matrix of $\C^\perp$, and vice versa.
	
	For a \emph{syndrome} $\uv \in \F_q^{n-k}$, the coset $\C_\uv = \{\yv \in \F_q^n : \Hm \yv^\top = \uv\}$ collects all words with syndrome $\uv$. Similarly, for a \emph{dual syndrome} $\uv \in \F_q^k$, the dual coset is $\C^\perp_\uv = \{\yv \in \F_q^n : \Gm \yv^\top = \uv\}$.
	
	The \emph{Hamming weight} of a vector $\yv \in \F_q^n$, written $\|\yv\|$, is the number of its nonzero coordinates; the \emph{Hamming distance} between $\yv$ and $\zv$ is $d_H(\yv, \zv) = \|\yv - \zv\|$. The \emph{minimum distance} of $\C$ is $d(\C) = \min_{\cv \neq \cv' \in \C} d_H(\cv, \cv')$.
	
	We next define bounded-distance decoding, which is the central decoding task throughout the paper:
	
	\begin{problem}[Bounded-distance decoding $\BDD(\C, t, \yv_0)$]\label{def:BDD}
		~\\ \textbf{Given:} A $q$-ary linear code $\C$ of dimension $k$ and length $n$ with parity-check matrix $\Hm \in \F_q^{(n-k) \times n}$, a received word $\yv_0 \in \F_q^n$, and a radius $t \in \Iint{0}{n}$. \\ \textbf{Goal:} Find a codeword $\cv \in \C$ such that $d_H(\cv, \yv_0) \le t$.
	\end{problem}
	
	\begin{remark}
		This is equivalent to finding $\yv \in \F_q^n$ with $\Hm \yv^\top = \Hm \yv_0^\top$ and $\norm{\yv} \le t$: given such $\yv$, the vector $\cv = \yv_0 - \yv$ satisfies $\cv \in \C$ and $d_H(\cv, \yv_0) = \norm{\yv} \le t$.
	\end{remark}
	
	The $\BDD$ instances considered in this paper are on Reed--Solomon codes. A bounded-distance decoder of radius $t$ for $\C$ is an algorithm $\mathrm{Dec}_t: \F_q^n \to \C \cup \{\mathrm{error}\}$ such that, for every $\yv_0 \in \F_q^n$ with $\operatorname{dist}(\yv_0, \C) \le t$, where $
dist(y,\mathcal{C}) := \min_{c \in \mathcal{C}} d_H(y,c),
$ the output $\mathrm{Dec}_t(\yv_0)$ is a codeword $\cv \in \C$ satisfying $d_H(\yv_0, \cv) \le t$. This is the promise version of $\BDD$. The received words produced by the Cheng--Wan reduction satisfy this promise by construction. We use this promise formulation for search problems; in complexity statements, decision BDD asks whether such a codeword exists.
	
	A quantum variant of this task, the Quantum Decoding Problem ($\QDP$), and an inhomogeneous variant ($\IBDD$), are defined in Section~\ref{sec:quantum-red} where they are used.
	
	\subsection{Reed--Solomon codes and duals}
	
	\begin{definition}[Reed--Solomon code ${\RS[n,k]}_q$]\label{def:RS_code}
		Let $\F_q$ be a finite field and let $\alpha_1, \dots, \alpha_n \in \F_q$ be pairwise distinct elements, called evaluation points. Let $k$ be an integer with $1 \le k \le n \le q$. The Reed--Solomon code $\RS[n,k]_q$ with evaluation points $(\alpha_1, \dots, \alpha_n)$ is the linear code
		\[
		\mathcal{C} = \left\{ \bigl(f(\alpha_1), \dots, f(\alpha_n)\bigr) : f(x) \in \F_q[x],\ \deg(f) < k \right\} \subseteq \F_q^n.
		\]
		It has minimum distance $d = n - k + 1$.
	\end{definition}
	
	
	
	When $\C = \RS[n,k]_q$ has full support, meaning $\{\alpha_i\}_{i \in \Iint{1}{n}} = \F_q$ (which forces $n = q$), the dual is an ordinary Reed--Solomon code: $\C^\perp = \RS[n, n-k]_q$, also with full support.
	
	\begin{definition}[Vandermonde matrices for Reed--Solomon codes]\label{def:vandermonde-RS}
		The Vandermonde generator matrix of $\mathcal{C} = \RS[n,k]_q$ with evaluation points $\alpha_1, \dots, \alpha_n \in \F_q$ is the $k \times n$ matrix
		\[
		\mathbf{G}
		\;=\;
		\begin{pmatrix}
			1              & 1              & \cdots & 1              \\
			\alpha_1       & \alpha_2       & \cdots & \alpha_n       \\
			\alpha_1^2     & \alpha_2^2     & \cdots & \alpha_n^2     \\
			\vdots         & \vdots         & \ddots & \vdots         \\
			\alpha_1^{k-1} & \alpha_2^{k-1} & \cdots & \alpha_n^{k-1}
		\end{pmatrix},
		\]
		so that $\mathcal{C} = \{\, \mathbf{u} \mathbf{G} : \mathbf{u} \in \F_q^k \,\}$.
		In the full-support case $n=q$, a Vandermonde parity-check matrix of $\RS[q,k]_q$ is the $(q-k) \times q$ matrix
		\[
		\mathbf{H}
		\;=\;
		\begin{pmatrix}
			1                & 1                & \cdots & 1                \\
			\alpha_1         & \alpha_2         & \cdots & \alpha_q         \\
			\alpha_1^2       & \alpha_2^2       & \cdots & \alpha_q^2       \\
			\vdots           & \vdots           & \ddots & \vdots           \\
			\alpha_1^{q-k-1} & \alpha_2^{q-k-1} & \cdots & \alpha_q^{q-k-1}
		\end{pmatrix}.
		\]
		It generates the dual code $\RS[q,q-k]_q$ and satisfies $\mathcal C=\{\cv\in\F_q^q:\mathbf H\cv^\top=0\}$.
	\end{definition}
	
	\subsection{Characters and Fourier analysis over finite fields}
	
	We follow the conventions of~\cite{chailloux2024quantumadvantagesoftdecoders}.
	
	\begin{definition}[Characters of $\F_q$]
		Let $q = p^s$ for a prime $p$ and an integer $s \ge 1$. The \emph{characters} of $\F_q$ are the functions $\chi_y : \F_q \to \mathbb C$ indexed by elements $y \in \F_q$, defined by
		\[
		\chi_y(x) \;\triangleq\; e^{\frac{2i\pi \,\mathrm{tr}(xy)}{p}},
		\quad\text{where}\quad
		\mathrm{tr}(a) \;\triangleq\; a + a^p + a^{p^2} + \cdots + a^{p^{s-1}}.
		\]
		The product $xy$ is multiplication in $\F_q$. We extend this to vectors $\xv, \yv \in \F_q^n$ by setting
		\[
		\chi_\yv(\xv) \;\triangleq\; \prod_{i=1}^n \chi_{y_i}(x_i) \;=\; e^{\frac{2i\pi\, \mathrm{tr}(\xv \cdot \yv)}{p}},
		\]
		where $\xv \cdot \yv = \sum_i x_i y_i$ is the standard inner product over $\F_q$.
	\end{definition}
	
	Note that for $\xv \in \F_q^k$, $\yv \in \F_q^n$, and $\Gm \in \F_q^{k \times n}$, we have $\xv \Gm \cdot \yv = \xv \cdot \yv \Gm^\top$, so $\chi_\yv(\xv \Gm) = \chi_{\yv \Gm^\top}(\xv)$.
	
	\begin{proposition}[Properties of characters]\label{prop:characters}
		The characters $\chi_\yv : \F_q^n \to \mathbb C$ satisfy:
		\begin{enumerate}
			\item \textup{(Group homomorphism.)} For all $\yv \in \F_q^n$, $\chi_\yv$ is a group homomorphism from $(\F_q^n, +)$ to $(\mathbb C^\times, \cdot)$: for all $\xv, \xv' \in \F_q^n$, $\chi_\yv(\xv + \xv') = \chi_\yv(\xv) \cdot \chi_\yv(\xv')$.
			\item \textup{(Symmetry.)} For all $\xv, \yv \in \F_q^n$, $\chi_\yv(\xv) = \chi_\xv(\yv)$.
			\item \textup{(Orthogonality.)} For all $\xv, \xv' \in \F_q^n$,
			\[
			\sum_{\yv \in \F_q^n} \chi_\yv(\xv)\,\overline{\chi_\yv(\xv')} = q^n \,\delta_{\xv,\xv'}.
			\]
			In particular, $\sum_{\yv \in \F_q^n} \chi_\yv(\xv) = 0$ for all $\xv \in \F_q^n \setminus \{\mathbf{0}\}$.
		\end{enumerate}
	\end{proposition}
	
	\begin{definition}[Fourier transform over $\F_q^n$]
		For a function $f : \F_q^n \to \mathbb C$, its \emph{Fourier transform} is the function $\hf : \F_q^n \to \mathbb C$ defined by
		\[
		\hf(\xv) \;\triangleq\; \frac{1}{\sqrt{q^n}} \sum_{\yv \in \F_q^n} \chi_\xv(\yv)\, f(\yv).
		\]
	\end{definition}
	
	The Fourier transform over $\F_q^n$ is essentially the quantum Fourier transform (QFT) over the additive group $(\F_q^n, +)$. Applied to a computational basis state $|\xv\rangle$, the QFT produces
	\[
	\mathrm{QFT}\, |\xv\rangle = \frac{1}{\sqrt{q^n}} \sum_{\yv \in \F_q^n} \chi_\yv(\xv)\, |\yv\rangle,
	\]
	which reduces to the Hadamard transform when $q = 2$.
	
	The Fourier transform satisfies Parseval's identity:
	
	\begin{claim}[Parseval's identity]
		For any $f : \F_q^n \to \mathbb C$, $\|f\|_2 = \|\hf\|_2$.
	\end{claim}
	
	We also record the following character sum identity, which relates the Fourier transform to the dual code and is used in the reduction of Section~\ref{sec:quantum-red}:
	
	\begin{claim}[Character sums over codes]\label{claim:char-sum}
		Let $\C$ be a $q$-ary linear code. Then
		\[
		\sum_{\cv \in \C} \chi_\yv(\cv) =
		\begin{cases}
			|\C| & \text{if } \yv \in \C^\perp, \\
			0 & \text{otherwise.}
		\end{cases}
		\]
	\end{claim}
	
	\begin{proof}
		If $\yv \in \C^\perp$ then $\chi_\yv(\cv) = e^{2i\pi\,\mathrm{tr}(\yv \cdot \cv)/p} = 1$ for all $\cv \in \C$, giving a sum of $|\C|$. Otherwise, $\F_q$-linearity implies that $\cv\mapsto\yv\cdot\cv$ maps $\C$ onto $\F_q$. The character $\chi_\yv$ therefore restricts to a nontrivial character on $(\C,+)$, and the sum of a nontrivial character over a finite abelian group is zero.
	\end{proof}
	
	\section{Discrete logarithm in finite abelian groups}
	\label{sec:DLOG_abelian}
	
	In this section, we formalize the discrete logarithm problem in a finite abelian group. We give a natural definition under which the problem decomposes into the corresponding discrete logarithm problems on the cyclic components. Under this formulation, average-to-worst-case hardness for discrete logarithm in cyclic groups extends directly to the finite abelian setting.
	
	\begin{definition}[Discrete logarithm in a finite abelian group with finite-field encoding]
		\label{def:DLOG-abelian-ff}
		Let $G = C_1 \times \cdots \times C_r$ be a finite abelian group written multiplicatively, where each $C_i = \langle \mathbf{b}_i \rangle$ is cyclic of known order $n_i$.
		
		Assume that for each $i$ we are given a finite field $\mathbb F_{p_i^{\ell_i}}$ and an explicit embedding
		$C_i \hookrightarrow \mathbb F_{p_i^{\ell_i}}^\times$.
		The elements of $C_i$ are represented by their standard encodings as elements of the ambient field $\mathbb F_{p_i^{\ell_i}}$.
		
		The discrete logarithm problem in $G$ relative to the basis
		$\mathbf{b} = (\mathbf{b}_1,\dots,\mathbf{b}_r)$ is the following task.
		
		\medskip
		
		\noindent\textbf{Input:}
		For each $i$, a description of the field $\mathbb F_{p_i^{\ell_i}}$, the embedded cyclic subgroup
		$C_i = \langle \mathbf{b}_i \rangle \le \mathbb F_{p_i^{\ell_i}}^\times$,
		the generator $\mathbf{b}_i$, and an element
		$\yv_i \in C_i$ given as a field element of $\mathbb F_{p_i^{\ell_i}}$.
		Equivalently, the input element is
		$\yv = (\yv_1,\dots,\yv_r) \in C_1 \times \cdots \times C_r$.
		
		\noindent\textbf{Output:}
		The unique tuple
		$(k_1,\dots,k_r) \in \mathbb Z/n_1\mathbb Z \times \cdots \times \mathbb Z/n_r\mathbb Z$
		such that $\yv_i = \mathbf{b}_i^{k_i}$ for every $i$. We denote this tuple by $\mathrm{DLOG}_G(\mathbf{b},\yv)$.
	\end{definition}

	This formulation makes the problem well defined. Since $G$ is given as an explicit direct product of cyclic groups, every element $\yv\in G$ admits a unique exponent vector $(k_1,\dots,k_r)$ relative to the basis $\mathbf{b}$. The usual discrete logarithm problem in a cyclic group is recovered as the special case $r=1$.
	
	The encoding is part of the problem definition. We assume that each cyclic factor is supplied with an explicit embedding into a finite-field multiplicative group. The divisibility condition $n_i\mid(p_i^{\ell_i}-1)$ guarantees that an abstract subgroup of order $n_i$ exists, but does not provide an efficient conversion from another encoding. In the generic-group model, elements have opaque encodings and the algorithm accesses them through group-operation oracles. The finite-field algorithms below use the explicit field representation.
	
	We recall the random self-reduction for the discrete logarithm problem in cyclic groups:
	
	\begin{theorem}[Average-to-worst-case reduction for discrete logarithm in cyclic groups]
		\label{thm:random-self-reduction-general}
		Let $C=\langle \mathbf g\rangle$ be a cyclic group of known order $n$. Suppose there is an efficient algorithm $\mathcal A$ such that, for uniformly random $\mathbf u\in C$,
		\[
		\Pr[\mathcal A(\mathbf u)=\mathrm{DLOG}_C(\mathbf g,\mathbf u)] = \epsilon.
		\]
		Then there is an efficient randomized algorithm $\mathcal B$ such that, for every $\mathbf u\in C$,
		\[
		\Pr[\mathcal B(\mathbf u)=\mathrm{DLOG}_C(\mathbf g,\mathbf u)] = \epsilon,
		\]
		where the probability is over the internal randomness of $\mathcal B$.
	\end{theorem}
	
	\begin{proof}
		The usual self-reduction~\cite[Thm.~10.2]{BonehShoup2020} also applies when the group order is composite. We include the proof for completeness. Fix $\mathbf u=\mathbf g^x\in C$. Algorithm $\mathcal B$ chooses $r\leftarrow \mathbb Z/n\mathbb Z$ uniformly at random and forms
		\[
		\mathbf u' := \mathbf u\,\mathbf g^r = \mathbf g^{x+r}.
		\]
		Since $r$ is uniform modulo $n$, the exponent $x+r \bmod n$ is uniform in $\mathbb Z/n\mathbb Z$, and therefore $\mathbf u'$ is uniform in $C$. Algorithm $\mathcal B$ then runs $\mathcal A$ on input $\mathbf u'$. If $\mathcal A$ outputs
		\[
		z=\mathrm{DLOG}_C(\mathbf g,\mathbf u')=x+r \pmod n,
		\]
		then $\mathcal B$ outputs
		$
		z-r \pmod n,
		$
		which equals $x=\mathrm{DLOG}_C(\mathbf g,\mathbf u)$.
		
		$\mathcal B$ succeeds exactly when $\mathcal A$ succeeds on the uniformly random element $\mathbf u'$, so
		$
		\Pr[\mathcal B(\mathbf u)=\mathrm{DLOG}_C(\mathbf g,\mathbf u)] = \epsilon.$ 
	\end{proof}
	
	In particular, the proof does not require $n$ to be prime; it holds for arbitrary finite cyclic groups.
	
	Computing the exponent vector reduces to solving the discrete logarithm problem in each cyclic factor:
	
	\begin{proposition}[Reduction to cyclic factors and average-to-worst-case transfer]
		\label{prop:DLOG-componentwise}
		Let $G = C_1 \times \cdots \times C_r$, where each $C_i = \langle \mathbf{b}_i \rangle$, and let $\yv = (\yv_1,\dots,\yv_r) \in G$. Then computing $\mathrm{DLOG}_G(\mathbf{b},\yv)$ is equivalent to computing, for each $i$, the cyclic discrete logarithm $\mathrm{DLOG}_{C_i}(\mathbf{b}_i,\yv_i)$, up to polynomial overhead.
		
		Moreover, by Theorem~\ref{thm:random-self-reduction-general}, inverse-polynomial average-case success on each component, with verifiable outputs, implies worst-case solvability of $\mathrm{DLOG}_G(\mathbf{b},\yv)$, again up to polynomial overhead.
	\end{proposition}
	
	\begin{proof}
		By definition,
		$\mathrm{DLOG}_G(\mathbf{b},\yv)$
		is the unique tuple
		$(k_1,\dots,k_r)$
		such that
		$\yv_i = \mathbf{b}_i^{k_i}$
		for every $i$.
		Solving the problem in $G$ amounts to solving the $r$ cyclic discrete logarithm instances independently.
		
		The final claim follows by applying Theorem~\ref{thm:random-self-reduction-general} to each cyclic factor separately, verifying and amplifying the success probability, and combining the resulting logarithms.
	\end{proof}

	\subsection{Classical algorithms for the discrete logarithm problem in abelian groups}
	
	We briefly review classical algorithms for the discrete logarithm problem, distinguishing finite-field algorithms from generic group algorithms.

			

\begin{table}[H] \centering \caption{Classical running times for the discrete logarithm problem in finite fields $\F_Q^\times$, where $Q$ is a prime power. The complexity is expressed in terms of the input length $\log Q$.} \label{tab:dlog-finite-fields} \small \renewcommand{\arraystretch}{1.15} \begin{tabularx}{\linewidth}{@{} l X l @{}} \toprule \textbf{Regime} & \textbf{Complexity} & \textbf{Reference} \\ \midrule Small characteristic $p\le(\log Q)^{O(1)}$ & $(\log Q)^{O(\log\log Q)}$ (unconditional expected time) & \cite{kleinjung2019discretelogarithmsquasipolynomialtime} \\ Medium/large characteristic & $\exp\!\bigl((c+o(1))(\log Q)^{1/3} (\log\log Q)^{2/3}\bigr)$  (heuristic) & \cite{JLSV06} \\ \bottomrule \end{tabularx} \end{table}

	
    In our setting $Q=q^h$. When $q\asymp h^4$, we have
    \[
    \log Q=h\log q,\qquad p\le q=(\log Q)^{O(1)},
    \]
    so this parameter choice falls in the small-characteristic regime relative to $Q$. 
    By Theorem~1.1 of Kleinjung and Wesolowski~\cite{kleinjung2019discretelogarithmsquasipolynomialtime}, discrete logarithms in these fields lie in the quasipolynomial-time regime.
    
    The bound $h\le q^{1/4}-2$ alone does not determine the characteristic regime when $q$ is allowed to be much larger. We therefore keep the dependence on the finite-field parameters explicit.

    	\begin{table}[H]
		\centering
		\caption{Heuristic expected classical running times for generic algorithms in cyclic factors of finite abelian groups. Let $G$ be a finite abelian group and let $\langle \mathbf g\rangle \le G$ be
			a cyclic subgroup of order $N$. Write
			$
			N \;=\; \prod_{j=1}^t p_j^{e_j}
			$
			for the prime factorisation of $N$, and let $
			p_{\max} \;:=\; \max_j p_j $
			be the largest prime divisor of $N$. The factorization of $N$ is assumed to be supplied, and complexity is measured in group operations and polynomial factors in $\log N$.}
		\label{tab:dlog-abelian-generic}
		\small
		\renewcommand{\arraystretch}{1.15}
		\begin{tabularx}{\linewidth}{@{} l X l @{}}\toprule
			\textbf{Regime} & \textbf{Complexity} & \textbf{Reference} \\
			\midrule
			All small prime divisors $p_j \le B\leq \poly(\log N)$
			& $\tilde{O}\big(\sqrt{B}\,\poly(\log N)\big)$ (polynomial)
			& \cite{PohligHellman78} \\
			
			Largest prime divisor $p_{\max}$
			& $\tilde{O}\big(\sqrt{p_{\max}}\,\poly(\log N)\big)$
			& \cite{PohligHellman78, Pollard78} \\
			\bottomrule
		\end{tabularx}
	\end{table}
	
    Shoup~\cite{Shoup97} gives a lower bound of $\Omega(\sqrt{p_{\max}})$ group operations in the generic-group model. This matches the leading dependence of the generic algorithms in Table~\ref{tab:dlog-abelian-generic}. These lower bounds do not apply to algorithms that exploit a finite-field encoding. In particular, a large prime divisor alone does not establish hardness against such algorithms.

    For the representation in Definition~\ref{def:DLOG-abelian-ff}, a hard cyclic factor makes the discrete logarithm problem in the product group hard. Solving the problem in the whole group nevertheless requires computing the logarithm in every factor. No assumption that one factor has order $\Theta(|G|)$ is needed.

    	\section{Classical reduction from discrete logarithm to decoding}\label{sec:classical-red}
	
	In this section, we revisit the classical reduction of Cheng and Wan~\cite{cheng2007list} and describe the bounded-distance decoding instances to which we will apply Regev's reduction. We also briefly review the known classical algorithms and hardness results for bounded-distance decoding of Reed--Solomon codes.
	
	\subsection{The Cheng--Wan reduction for finite abelian groups}
	
	We now recall the Cheng--Wan reduction~\cite{cheng2007list}, which relates the discrete logarithm problem in $\F_{q^h}^\times$ to decoding Reed--Solomon received words. For our purposes, it suffices to isolate three ingredients: the construction of the induced received word $\yv^{(i)}$, the fact that this yields a promised BDD instance, and the relation-collection step obtained from successful decoding. For any prime power $q$ and any integer $h$ satisfying
	$
	2 \le h \le q^{1/4}-2,
	$
	one has the randomized reduction
	\[
	\mathrm{DLOG}_{\F_{q^h}^\times}(\mathbf{b},\cdot)
	\;\le_{\mathrm{rand}}\;
	\BDD^{\mathrm{promise}}(\RS[q,3h+4]_q,q-4h-4).
	\]

	We first describe the received-word construction.
	
	\begin{definition}[Cheng--Wan received word]
		\label{def:noisy_codeword}
		Fix a monic irreducible $h(x)\in\F_q[x]$ of degree $h>1$ and set
		$\alpha := x \bmod h(x)$ so that $\F_{q^h}\cong\F_q(\alpha)$.  Let \(\mathbf{b}\) be a primitive generator of \(\F_{q^h}^\times\), and let \(b(\alpha) \in \F_q[\alpha]\) denote its corresponding representative, where \(b(x) \in \F_q[x]\) is the unique polynomial of degree \(< h\).
		
		For any integer $i \ge 0$, define
		\[
		f(x) := b(x)^i \bmod h(x)\qquad(\deg f < h),
		\]
		and let $\mathcal{C} := \RS[q,3h+4]_q$ be the full-support Reed--Solomon code of
		length $q$ and dimension $3h+4$ over $\F_q$.
		We define the received word induced by $\mathbf{b}$ and $i$ as
		\[
		\yv^{(i)} := \bigl( y_a^{(i)} \bigr)_{a\in\F_q}
		\quad\text{with}\quad
		y_a^{(i)} := -\,\frac{f(a)}{h(a)} - a^{3h+4},
		\]
		which is well-defined since $h(a)\neq 0$. $\yv^{(i)}$ lies in the ambient space $\F_q^q$ of $\mathcal{C}$.
	\end{definition}

	By \cite[Theorem~4]{cheng2007list}, when \(q \ge (h+2)^4\), the element \(f(\alpha)\in\F_{q}[\alpha]\) can be written as
	\(f(\alpha)=\prod_{a\in A}(\alpha-a)\)
	for some subset \(A\subseteq S=\F_q\) of size \(|A|=4h+4\). Writing
	\(P(x):=\prod_{a\in A}(x-a)\),
	the identity \(P(\alpha)=f(\alpha)\) and the minimality of \(h(x)\) imply that \(h(x)\mid(P(x)-f(x))\), so
	\[
	P(x)=f(x)+t(x)h(x)
	\]
	for some \(t(x)\in\F_q[x]\) that is monic of degree \(3h+4\). Setting
	\(u^\star(x):=t(x)-x^{3h+4}\),
	one finds that \(u^\star(a)=y_a^{(i)}\) for every \(a\in A\). The codeword \(\cv^\star=(u^\star(a))_{a\in S}\in\mathcal C\) agrees with \(\yv^{(i)}\) in at least \(4h+4\) coordinates. Hence
	\[
	\dist(\yv^{(i)},\mathcal C)\le n-(4h+4).
	\]
	Therefore \(\yv^{(i)}\) is a promise bounded-distance decoding instance.
	
	Conversely, let $u(x)$ be the degree-$<3h+4$ polynomial returned by a valid decoder and set $t(x)=x^{3h+4}+u(x)$. The polynomial $f(x)+t(x)h(x)$ is monic of degree $4h+4$, and each agreement between the decoded codeword and the received word is a root. There are at least $4h+4$ such roots, so there are exactly $4h+4$ and
\[
 f(x)+t(x)h(x)=\prod_{a\in A_i}(x-a).
\]
Evaluating at $\alpha$ gives the multiplicative relation
\[
 b(\alpha)^i=\prod_{a\in A_i}(\alpha-a).
\]
This also proves that each Cheng--Wan received word has distance exactly $q-4h-4$ from the code.

Write $N=q^h-1$ and enumerate the factor base as $a_1,\ldots,a_q$, where each $a_j$ is an element $\alpha-a$ with $a\in\F_q$. Let $l_j=\log_b(a_j)\pmod N$. A factorization of $b^i$ gives a row $v\in\{0,1\}^q$ and the equation $v^\top l=i\pmod N$. Since $N$ need not be prime, the relation rows must generate $(\mathbb Z/N\mathbb Z)^q$ as a module.

\begin{corollary}[Relation collection in the Cheng--Wan reduction]\label{cor:cheng-wan-augmented-relations}
Suppose a decoder succeeds with probability $\epsilon>0$ on the received word associated with a uniformly random element of $\F_{q^h}^\times$. For any $0<\delta<1$, one can collect
\[
 J=O\!\left(q\log\frac{q\log N}{\delta}\right)
\]
successful ordinary and shifted relations whose rows generate $(\mathbb Z/N\mathbb Z)^q$ with probability at least $1-\delta$. The expected number of decoder calls is $J/\epsilon$.
\end{corollary}
\begin{proof}
For an ordinary relation, sample $i$ uniformly modulo $N$, decode $b^i$, and retain the equation $v^\top l=i$ when decoding succeeds. For a shifted relation in coordinate $j$, decode $b^i/a_j$ and retain $(e_j+v)^\top l=i$. In both cases the element being decoded is uniform in $\F_{q^h}^\times$. The conditional distribution of $v$ given successful decoding is the same, even if the decoder's success depends on the input. Lemma~\ref{lem:pomerance-relations}, applied to this distribution, gives the generation bound. Each successful row takes $1/\epsilon$ calls in expectation.
\end{proof}

Once the rows generate the module, all factor-base logarithms are determined. They can be recovered in polynomial bit complexity using Hermite or Smith normal form; factoring $N$ is not required. For a target $y\in\F_{q^h}^\times$, choose a uniform $r$ modulo $N$ and decode $yb^r$. A successful factorization gives
\[
 \log_b y=v^\top l-r\pmod N.
\]
This final step costs another $1/\epsilon$ decoder calls in expectation.

\begin{theorem}[Discrete logarithms via Reed--Solomon decoding]\label{thm:abelian-DLOG-BDD}
Let $2\le h\le q^{1/4}-2$. Assume an explicit representation of $\F_{q^h}$ and a supplied primitive generator $b$ of its multiplicative group. Suppose a decoder runs in time $T_D$ and succeeds with probability $\epsilon>0$ on the Cheng--Wan received words associated with uniform field elements. Then, for every target, its logarithm to base $b$ can be computed with probability at least $1-\delta$ in expected time
\[
 O\!\left(\frac{(J+1)(T_D+T_{\rm CW})}{\epsilon}\right)
 +\poly(q,\log N,\log(1/\delta)),
\]
where $N=q^h-1$, $J$ is as in Corollary~\ref{cor:cheng-wan-augmented-relations}, and $T_{\rm CW}$ is the cost of field exponentiation, received-word construction and output validation for one call. In particular, this is polynomial in $q$ if $T_D=\poly(q)$ and $\epsilon\ge1/\poly(q)$.

The reduction extends to the groups in Definition~\ref{def:DLOG-abelian-ff}, provided each ambient field satisfies the parameter condition and a primitive generator of each ambient multiplicative group is supplied. The total running time is the sum of the component costs, up to polynomial overhead.
\end{theorem}
\begin{proof}
The finite-field statement follows from the received-word construction, Corollary~\ref{cor:cheng-wan-augmented-relations}, and the target randomization above. Successful outputs are checked before recording a relation.

For a cyclic subgroup $C=\langle c\rangle\le\F_{q^h}^\times$ and a target $y\in C$, first compute $a=\log_b c$ and $z=\log_b y$. Solving $ax\equiv z\pmod N$ gives the unique solution modulo $|C|$. Applying this procedure to each cyclic factor proves the statement for the product group.
\end{proof}

The running time is polynomial in the input bit length $h\log q$ when $q$ is polynomially bounded in that length, as happens for $q\asymp h^4$.

Cheng and Wan~\cite{cheng2008complexity} also extend the Cheng--Wan reduction to the positive-rate setting. Their reduction is essentially unchanged, except that the underlying factorization is padded to cardinality
	$
	g=(c+o(1))q,
	$
	for some fixed \(c\in(0,1)\), by augmenting a factorization over a suitable subfield with additional distinct linear factors from \(\F_q\setminus \F_{q_1}\), where \(q_1\) denotes an auxiliary proper subfield used in the construction. Thus
	\[
	b(\alpha)^i \;=\; \prod_{a\in A_i}(\alpha-a),
	\qquad |A_i|=g.
	\]
	Accordingly, the low-rate code \(\RS[q,3h+4]_q\) is replaced by the positive-rate code \(\RS[q,g-h]_q\), while the extracted relation vector and the subsequent linear-algebra step remain the same. Moreover, since the specially constructed received words lie at exact distance \(q-g\) from the code, a maximum-likelihood decoding oracle returns precisely the codewords yielding the same multiplicative relations as in the bounded-distance setting.
	
	Therefore, in this positive-rate regime, the discrete logarithm problem can likewise be solved in time \(q^{O(1)}\) by applying a polynomial-time maximum-likelihood decoding algorithm to the noisy codewords produced by the Cheng--Wan reduction.
	
	We will use the following notation for the induced distribution of received words:
	
	\begin{definition}[Cheng--Wan instance distribution $\mathcal D_{q,h,g}$]
    \label{def:CW-distribution}
    Fix $q$, $h\ge2$ and $g$, set $n=q$ and $k=g-h$, and let $h(x)\in\F_q[x]$ be monic irreducible of degree $h$. We identify $\F_{q^h}$ with $\F_q[x]/(h(x))$. Fix a primitive generator $b$ of $\F_{q^h}^\times$ and sample $i$ uniformly from $\{0,\ldots,q^h-2\}$.

    Let $f_i(x)$ be the unique polynomial of degree $<h$ representing $b^i$, and define
    \[
    y^{(i)}(a)=-\frac{f_i(a)}{h(a)}-a^{g-h},\qquad a\in\F_q.
    \]
    We denote the distribution of $\yv^{(i)}$ by $\mathcal D_{q,h,g}$.
\end{definition}

The low-rate construction uses $g=4h+4$. The positive-rate construction has $g=(c+o(1))q$ for a fixed $c\in(0,1)$, under the additional parameter conditions of~\cite{cheng2008complexity}.

\begin{definition}[Cheng--Wan BDD instance distribution $\mathcal D^{\BDD}_{q,h,g}$]
    \label{def:CW-BDD-distribution}
    With the parameters and received word of Definition~\ref{def:CW-distribution}, let $\C=\RS[q,g-h]_q$ and let $\Hm\in\F_q^{(q-g+h)\times q}$ be a parity-check matrix of $\C$. The corresponding syndrome is
    \[
    \mathbf u^{(i)}=\Hm(\yv^{(i)})^\top.
    \]
    We write $\mathcal D^{\BDD}_{q,h,g}$ for the distribution of the BDD instance $(\C,q-g,\yv^{(i)})$ induced by the uniform choice of $i$.
\end{definition}

\subsection{Classical algorithms for decoding Reed--Solomon codes }\label{sec:classical_decoders}
	
	The reduction from $\DLOG$ uses a specialized
	promise bounded-distance decoding instance.
	More precisely, the reduction produces instances of promise BDD for the Reed--Solomon
	family $\RS[q,k]_q$, for the relevant choice of $k$, and only on the structured
	inputs induced by the reduction itself.
	
	Because this is a promise formulation, the existence of a codeword within the prescribed radius is guaranteed by construction. For the specific Cheng--Wan instances, the received word has exact distance equal to the target decoding radius, so bounded-distance decoding and maximum-likelihood decoding coincide on these instances. There may be several codewords at that exact distance, and any such codeword is a valid output for the reduction.
	
	This promise task is more specialized than the standard worst-case decoding problems studied in coding theory: the guarantee of a nearby codeword comes from the promise on the input, not from a worst-case radius condition over all received words. We relate this task to standard worst-case decoding and list decoding below.
	
	The Cheng--Wan promise BDD task for a fixed code \(\RS[n,k]_q\) is no harder than worst-case promise BDD, where the code description is part of the input and the decoder must handle arbitrary codes of the same blocklength and radius. An algorithm for the corresponding unpromised search problem also suffices: any algorithm that outputs a codeword within distance $t$ whenever one exists automatically solves the promised variant by restriction to inputs satisfying the promise. We therefore focus on the latter, more standard formulations.
	
	On the algorithmic side, any polynomial-time list decoder of radius $t$ solves promise BDD at that radius: run the decoder and return any codeword in its nonempty output list. When $t\le\lfloor(d_{\min}-1)/2\rfloor$, the list contains at most one codeword, but uniqueness is not required by our BDD formulation.

	For Reed--Solomon codes one can go beyond half the minimum distance via the
	Guruswami--Sudan algorithm. For $\RS[n,k]_q$ of rate
	$R = k/n$, Guruswami--Sudan list decodes in polynomial time up to
	\[
	t_{\mathrm{GS}} \approx n(1-\sqrt{R})
	= n\Bigl(1-\sqrt{\tfrac{k}{n}}\Bigr)
	\]
	errors. The decoding ball may contain multiple codewords; returning any element of the list solves promise BDD at the same radius.
	
	On the complexity-theoretic side, we state the following hardness results for BDD:

	\paragraph{NP-completeness at constant rate.} In the decision setting, bounded-distance decoding for Reed--Solomon codes
	is known to be $\NP$-hard at the following radii.
	First, Guruswami and Vardy~\cite{GuruswamiVardy2005} imply that
	RS BDD is $\NP$-complete at radius
	$n-k-1$. Second, Gandikota, Ghazi, and Grigorescu~\cite{GandikotaGhaziGrigorescu2018}
	showed that this hardness extends to smaller radii: for constant-rate
	Reed--Solomon codes, the decision problem remains $\NP$-complete at radius
	$n-k-d$ for every
	$1\le d\le c\log n/\log\log n$ for a sufficiently small constant $c>0$, and is already
	NP-hard under quasi-polynomial reductions for $d=\Theta(\log n)$.
	
	\paragraph{NP-completeness for asymptotically zero-rate Reed--Solomon codes.}
	
	We extend the finite-field construction of~\cite{GandikotaGhaziGrigorescu2018} to Reed--Solomon codes of asymptotically zero rate. An instance of Moments Subset Sum, denoted MSS$(d)$, consists of a set $A$ in a field, a required subset size $K$, and target moments $m_1,\ldots,m_d$. It asks whether some $S\subseteq A$ of size $K$ satisfies $\sum_{a\in S}a^j=m_j$ for every $1\le j\le d$.

    We pad the MSS instance within a field extension and then follow the reduction chain
    \[
    \text{1-in-3SAT}\;\longrightarrow\;\text{MSS}\;\longrightarrow\;\text{SSS}\;\longrightarrow\;\text{RS BDD},
    \]
    where SSS specifies elementary symmetric sums. The padding makes the Reed--Solomon rate tend to zero. Appendix~\ref{appendix:Subset-sum} states the auxiliary reductions.

    \begin{lemma}[Size padding for MSS$(d)$ over a finite field, with fixed subset size]
\label{lem:mss_padding_fixed_k}
Let an MSS instance over a finite field \(F\) consist of a set
\(A\subset F\), \(|A|=N\), a required subset size \(K\), and target
moments \(m_1,\ldots,m_d\), with \(d\ge1\).
Let \(M=N^4-N\), and let \(E/F\) be an extension of degree \(M+1\)
with power basis \(1,\theta,\ldots,\theta^M\).
For \(N\ge2\), set
\[
 A'=A\cup\{\theta,\ldots,\theta^M\}\subset E.
\]
Retain \(K\) and all target moments, using the natural embedding
\(F\subset E\). The padded instance is YES if and only if the original
instance is YES, and \(|A'|=N^4\).
\end{lemma}

\begin{proof}
An original solution remains a solution. Conversely, in the first-moment
identity for any solution \(S'\subset A'\),
\[
 \sum_{x\in S'}x=m_1\in F,
\]
the coefficient of \(\theta^j\) is exactly \(1\) if its corresponding
dummy element was selected and \(0\) otherwise. Basis independence
forces every such coefficient to vanish, so no dummy is selected.
The argument uses subsets of distinct elements, so it works in every
characteristic; each dummy contributes either $0$ or $1$ to its own basis coordinate.
\end{proof}

\begin{theorem}[NP-completeness of Reed--Solomon BDD at asymptotically zero rate]
\label{thm:ggg18}
There is a family of explicitly represented prime-power fields \(E\)
and Reed--Solomon instances over \(E\), with \(\log|E|\) polynomially
bounded in their block length \(n\), for which decision BDD at radius
\(n-k-d\) is NP-complete in the polynomial-reduction range
\[
 1\le d\le c\,\frac{\log n}{\log\log n}
\]
for a sufficiently small absolute constant \(c>0\), and
\[
 k=\Theta(n^{1/4}),\qquad \frac{k}{n}=\Theta(n^{-3/4}).
\]
The corresponding \(d=\Theta(\log n)\) hardness statement holds under
quasi-polynomial-time reductions, with the associated
quasi-polynomial encoding bounds.
\end{theorem}

\begin{proof}
Use the finite-field MSS reduction in Lemma~\ref{lem:13sat_to_MSS_Q}. For a
1-in-3SAT instance with \(v\) variables, it supplies \(N=v(2^{d+1}-2)\)
distinct elements, with required size \(K=N/2\), in a field of
characteristic \(p>d\); the characteristic and encoding size are
bounded by a polynomial in \(v,d!\). In particular one can use
the characteristic choice \(p=O(d!)\) from that construction,
with a constant adjustment for the smallest values of \(d\).
In the polynomial-reduction range, \(d!\) is polynomial in \(v\).

Apply Lemma~\ref{lem:mss_padding_fixed_k}, translate as in Lemma~\ref{lem:Q_to_Z}, and reduce to Reed--Solomon BDD by Lemma~\ref{lem:Z_to_Fq}. This gives
\[
 n=N^4+1,\quad k=N/2-d+1,\quad t=N^4-N/2.
\]
Since \(d=o(N)\), the claimed dimension and rate follow.
The padding extension multiplies the bit length of field elements by
\(N^4-N+1\), still polynomial in the input size in the stated range.
The relation \(\log n=4\log N+o(1)\) only changes the constant in the
allowable \(\log n/\log\log n\) range. The same polynomial padding
preserves the quasi-polynomial reduction when \(d=\Theta(\log n)\).

For deterministic construction, use an irreducible-polynomial
algorithm with running time polynomial in the extension degree,
the characteristic \(p\), and \(\log|F|\), such as Shoup's
small-characteristic construction~\cite{Shoup90}.
Here \(p\) is polynomially bounded in the polynomial-reduction range,
so both field constructions are deterministic polynomial time.

Finally, a degree-\(<k\) polynomial is a polynomial-size certificate;
evaluation and distance checking are polynomial in the explicit
field representation. This establishes membership in NP.
\end{proof}

This family differs from the full-support Cheng--Wan family: its alphabet size need not equal its block length.

    \noindent The decision version of maximum-likelihood decoding is also $\NP$-complete for the Reed--Solomon families of Guruswami and Vardy~\cite{GuruswamiVardy2005}: membership in $\NP$ is immediate, since a candidate polynomial \(f(x)\in\F_q[x]\) of degree \(<k\) provides a certificate whose corresponding codeword can be checked in polynomial time against the received word, and their reduction from 3-Dimensional Matching proves $\NP$-hardness.
	
	In the Cheng--Wan reduction, however, the received words lie at a prescribed exact distance from the code, so bounded-distance decoding at that radius and maximum-likelihood decoding coincide; it therefore suffices to consider BDD.
	
	The following schematic compares the decoding radii. The hardness statements concern the respective code families described above.
	
	\noindent\resizebox{\linewidth}{!}{%
\begin{tikzpicture}[>=stealth, x=1.15cm, y=0.85cm, scale=1.15, transform shape]
		\draw[->] (-0.5,0) -- (10.10,0)
		node[right,align=center] {\scriptsize\shortstack{\setlength{\baselineskip}{7pt}BDD $\RS[n,k]_q$\\radius}};
		
		\draw (0,0) -- (0,0.15) node[below=2pt] {\scriptsize $0$};
		
		\draw (1.5,0) -- (1.5,0.15) node[below=2pt] {\scriptsize $n-k-(\sqrt{nk}-k)$};
		
		\draw (4.0,0) -- (4.0,0.15) node[below=2pt] {\scriptsize $n-k-h$};
		
		\draw (6.3,0) -- (6.3,0.15) node[below=2pt] {\scriptsize $n-k-c\frac{\log n}{\log\log n}$};
		
		\draw (8.1,0) -- (8.1,0.15) node[below=2pt] {\scriptsize $n-k-1$};
		
		\draw (9.3,0) -- (9.3,0.15) node[below=2pt] {\scriptsize $n-k$};
		\node[above=1pt,align=center] at (9.3,0.18)
		{\scriptsize\shortstack{\setlength{\baselineskip}{7pt}covering\\radius}};
		
		\draw[very thick] (0,0) -- (1.5,0);
		\node[above,align=center] at (0.8,0.22)
		{\scriptsize Poly-time\footnotemark[1]};
		
		\fill[red] (4.0,0) circle (2pt);
		\node[above,align=center, red] at (4.0,0.22)
		{\scriptsize DLOG-hard\footnotemark[2]};

		\draw[very thick, red] (6.3,0) -- (8.1,0);
		\node[above,align=center, red] at (7.3,0.22)
		{\scriptsize $\NP$-hard\footnotemark[3]};
	\end{tikzpicture}}
	
	\footnotetext[1]{Due to Guruswami--Sudan.}
	\footnotetext[2]{Under randomized polynomial-time reductions, following
		\cite{cheng2007list,cheng2008complexity} for the full-length case $n=q$.
		For low-rate codes, $h=\Theta(k)$; for positive-rate codes,
		$h=q^{1/4+o(1)}=o(k)$. MLD and promise BDD coincide on these received words.}
	\footnotetext[3]{Under polynomial-time reductions, following
		\cite{GandikotaGhaziGrigorescu2018} and the present work for the appropriate
		dependencies between $n$ and $q$. Here $c>0$ denotes the implicit constant 
		from Theorem~\ref{thm:ggg18}.}

	\section{Quantum reduction}\label{sec:quantum-red}

	We first define classical and quantum decoding problems with a prescribed error distribution.
	
	\begin{problem}[Decoding Problem $\DP(\C,p)$]
		~\\ \textbf{Given:} $(\C,p,\cv + \ev)$ where $\C$ is a $q$-ary linear code of dimension $k$ and length $n$, $\cv \Unif \C$, and $\ev \Unif p$ with $p : \F_q^n \rightarrow [0;1]$ satisfying $\sum_{\xv \in \F_q^n} p(\xv) = 1$. \\ \textbf{Goal:} Find $\cv$.
	\end{problem}
	
	\begin{problem}[Quantum Decoding Problem $\QDP(\C,f)$]
		~\\ \textbf{Given:} $(\C,f,\ket{\psi_\cv})$ where $\C$ is a $q$-ary linear code of dimension $k$ and length $n$, $\cv \Unif \C$, and $\ket{\psi_\cv} = \sum_{\ev \in \F_q^n} f(\ev) \ket{\cv + \ev}$ for a function $f : \F_q^n \rightarrow \mathbb{C}$ with $\norm{f}_2 = 1$. \\ \textbf{Goal:} Find $\cv$.
	\end{problem}

     Although \(\QDP(\C,f)\) is defined as the task of recovering \(\cv\), our algorithm requires a coherent implementation of the decoder that writes its estimate of \(\cv\) to an auxiliary register. 

	Measuring the state $\ket{\psi_{\cv}}$ in the computational basis gives a $\DP(\C,|f|^2)$ instance, so any $\DP(\C,p)$ solver can be used as a $\QDP(\C,f)$ solver for any $f$ such that $|f|^2 = p$.
	
	The $\BDD$ problem was already defined in Problem~\ref{def:BDD}. We call the variant with a uniformly random received word $\yv_0$ inhomogeneous BDD.
	
	\begin{problem}[Inhomogeneous Bounded-Distance Decoding $\IBDD(\C,t)$]
		~\\ \textbf{Given:} A $q$-ary linear code $\C$ of dimension $k$ and length $n$, a uniformly random starting point $\yv_0 \Unif \F_q^n$, and a radius $t \in \Iint{0}{n}$. \\ \textbf{Goal:} Find $\cv \in \C$ such that $d_H(\cv,\yv_0) \le t$, or equivalently find $\yv \in \F_q^n$ with $\Hm\yv^\top = \Hm\yv_0^\top$ and $\norm{\yv} \le t$, where $\Hm \in \F_q^{(n-k) \times n}$ is a parity-check matrix of $\C$.
	\end{problem}
	
	More generally, one can replace the Hamming-ball constraint with an arbitrary target set $T \subseteq \F_q^n$. $\IBDD(\C, T)$ asks, given a uniformly random $\yv_0$, to find $\yv \in \F_q^n$ such that $\Hm\yv^\top = \Hm\yv_0^\top$ and $\yv \in T$.
	
	\subsection{Main algorithmic reductions}
	
	We first recall the reduction with a perfect QDP solver; see~\cite{chailloux2025opixsoftdecoders} for more details.
	
	\begin{proposition}\label{prop:perfect-QDP}
		Let $\C$ be a $q$-ary linear code of dimension $k$ with parity-check matrix $\Hm\in\F_q^{(n-k)\times n}$, and let $(\C,t,\yv_0)$ be a $\BDD$ instance. Let $f : \F_q^n \rightarrow \mathbb{C}$ be normalized, with an efficiently preparable amplitude state, and suppose that $\mathrm{Supp}(\hf) \subseteq \{\yv \in \F_q^n : \norm{\yv} \le t\}$. If we have an efficient perfect quantum algorithm for $\QDP(\C^\perp,f)$ then we can efficiently solve $\BDD(\C,t,\yv_0)$.
	\end{proposition}
	\begin{proof}
		Let $\uv = \Hm\yv_0^\top \in \F_q^{n-k}$ be the syndrome of $\yv_0$. We start by constructing the state
		$$\frac{1}{\sqrt{q^{n-k}}} \sum_{\ev \in \F_q^n} f(\ev)\ket{\ev} \sum_{\sv \in \F_q^{n-k}} \overline{\chi_{\sv}(\uv)} \ket{\sv} .$$
		We then apply the unitary mapping $\ket{\ev}\ket{\sv} \rightarrow \ket{\ev + \sv\Hm}\ket{\sv}$ to obtain
		$$\frac{1}{\sqrt{q^{n-k}}}\sum_{\sv \in \F_q^{n-k}} \overline{\chi_{\sv}(\uv)}\ket{\psi_{\sv}}\ket{\sv}, \quad \text{with} \quad \ket{\psi_{\sv}} = \sum_{\ev \in \F_q^n} f(\ev)\ket{\sv\Hm + \ev}.$$
		Note that $\Hm$ is also a generator matrix of $\C^{\perp}$ so recovering $\sv$ from $\ket{\psi_{\sv}}$ is a  $\QDP(\C^\perp,f)$ instance. A perfect quantum algorithm for $\QDP(\C^\perp,f)$ therefore provides a unitary
		$$U : \ket{\psi_\sv}\ket{0} \rightarrow \ket{\psi_{\sv}}\ket{\sv}.$$
		Applying $U^\dagger$ yields $\frac{1}{\sqrt{q^{n-k}}}\sum_{\sv \in \F_q^{n-k}} \overline{\chi_{\sv}(\uv)} \ket{\psi_{\sv}}$. Finally, applying the quantum Fourier transform and using $\overline{\chi_\sv(\uv)} = \chi_\sv(-\uv)$ together with $\chi_\yv(\sv\Hm) = \chi_\sv(\Hm\yv^\top)$ gives
		$$\sqrt{q^{n-k}}\sum_{\yv : \Hm\yv^\top = \uv} \hf(\yv) \ket{\yv}.$$
		A perfect decoder distinguishes all the states $\ket{\psi_{\sv}}$, so they are orthogonal. Fourier inversion of their Gram matrix gives $\sum_{\Hm\yv^\top=\uv}|\hf(\yv)|^2=q^{-(n-k)}$ for every syndrome $\uv$, which also verifies the normalization above. The support condition implies that measuring this state yields an error $\yv$ of weight at most $t$; the codeword $\yv_0-\yv$ solves the BDD instance.
	\end{proof}

	The support condition can be relaxed. However, replacing a perfect QDP decoder by one that succeeds with probability close to $1$ can make the reduction fail on a fixed BDD instance~\cite{chailloux2024quantumadvantagesoftdecoders}. For a uniformly random received word, the following theorem gives a success guarantee and also allows Fourier mass outside the target set.
	
	\begin{theorem}[{\cite{chailloux2025opixsoftdecoders}}]\label{Theorem:Main}
		Let $\C$ be a $q$-ary linear code with parity-check matrix $\Hm$, and let $(\C, T)$ be an $\IBDD$ instance. Let $f : \F_q^n \rightarrow \mathbb{C}$ with $\norm{f}_2 = 1$. Assume that:
		\begin{enumerate}\setlength\itemsep{-0.3em}
			\item We have a quantum algorithm $\aa_{\QDP}$ that solves $\QDP(\C^\perp,f)$ in time $\Time_{\QDP}$ and succeeds with probability $P_{\mathrm{Dec}}$.
			\item The state $\sum_{\ev \in \F_q^n} f(\ev) \ket{\ev}$ is constructible in time $\Time_{\mathrm{Sampl}}$.
			\item $\sum_{\yv \in T}|\hf(\yv)|^2 = 1 - \eta$.
		\end{enumerate}
		Then there exists a quantum algorithm that solves $\IBDD(\C,T)$ with probability
		$$P \ge P_{\mathrm{Dec}}(1-\eta) - 2\sqrt{\eta P_{\mathrm{Dec}}(1-P_{\mathrm{Dec}})},$$
		and runs in time
		$$\Time = O\!\left(\frac{1}{P_{\mathrm{Dec}}} \left(\Time_{\QDP} + \Time_{\mathrm{Sampl}}\right) + \poly(n,\log(q))\right).$$
	\end{theorem}

    \subsection{Instantiation for Reed--Solomon codes in the Hamming metric}
	
	We now instantiate the reduction for product errors and Hamming balls. 
	
	\begin{definition}
		We define $\QDP(\C,\tau)$ to be the problem $\QDP(\C,f)$ where $f = u^{\otimes n}$ and $u$ is defined 
		\[ u(\alpha) = \left\{ \begin{array}{cl} \sqrt{1-\tau} & \textrm{if } \alpha = 0, \\ \sqrt{\dfrac{\tau}{q - 1}} & \textrm{if } \alpha \neq 0. \end{array}\right. \]
	\end{definition}
	
	\begin{definition}
		For a fixed $q$ and $\tau\in[0,1-1/q]$, define
		$$ \tau^{\perp} = \left(\sqrt{\frac{(q-1)(1-\tau)}{q}} - \sqrt{\frac{\tau}{q}}\right)^2.$$
		Moreover, one can check that $(\tau^\perp)^\perp = \tau$.
	\end{definition}
	
	With the above choice of $u$, we have 
	\[ \hu(\alpha) = \left\{ \begin{array}{cl} \sqrt{1 - \tau^\perp} & \textrm{if } \alpha = 0, \\ \sqrt{\dfrac{\tau^\perp}{q - 1}} & \textrm{if } \alpha \neq 0. \end{array}\right. \]

	We set $n=q$ and consider the full-support Reed--Solomon code $\RS[q,k]_q$, with $1\le k<q$.

	\begin{theorem}[Reed--Solomon instantiation]\label{thm:RS-instantiation}
        Suppose an algorithm for $\QDP(\RS[q,q-k]_q,\tau)$ succeeds with probability $P_{\mathrm{Dec}}>0$ in time $T_Q$. For an integer $1\le g\le q$, set
        \begin{equation}
        p_0=1-\tau^\perp
        =\frac{1+(q-2)\tau+2\sqrt{(q-1)\tau(1-\tau)}}q,
        \qquad m=qp_0,
        \label{eq:p0}
        \end{equation}
        and let
        \begin{equation}
        \eta_g=\Pr[Z<g],\qquad Z\sim\operatorname{Bin}(q,p_0).
        \label{eq:eta}
        \end{equation}
        There is an algorithm for $\IBDD(\RS[q,k]_q,q-g)$ with success probability at least
        \begin{equation}
        \max\{0,\ P_{\mathrm{Dec}}(1-\eta_g)-2\sqrt{\eta_gP_{\mathrm{Dec}}(1-P_{\mathrm{Dec}})}\}
        \label{eq:success}
        \end{equation}
        and running time
        \begin{equation}
        T_R=O\!\left(\frac{T_Q+T_{\rm prep}}{P_{\mathrm{Dec}}}+\poly(q,\log q)\right),
        \qquad T_{\rm prep}=\poly(q,\log q).
        \label{eq:reductiontime}
        \end{equation}
        Success is averaged over a uniformly random received word.

        In particular, with
        \begin{equation}
        L=\ln(64/P_{\mathrm{Dec}}),\qquad
        m_*=g+L+\sqrt{L^2+2gL},
        \label{eq:margin}
        \end{equation}
        the condition $m\ge m_*$ guarantees success at least $47P_{\mathrm{Dec}}/64$. If $m_*\le q$, it suffices to take $\tau\ge\tau_*$ within $[0,1-1/q]$, where
        \begin{equation}
        \tau_* = \frac1q\left(\sqrt{\frac{(q-1)m_*}{q}}-\sqrt{1-\frac{m_*}{q}}\right)^2.
        \label{eq:taustar}
        \end{equation}
    \end{theorem}
    \begin{proof}
        Apply Theorem~\ref{Theorem:Main} with the product amplitude $f=u^{\otimes q}$. Its state is efficiently preparable. After Fourier transformation, each coordinate is independently zero with probability $p_0=1-\tau^\perp$.

        A vector lies in the Hamming ball of radius $q-g$ exactly when at least $g$ coordinates are zero. The Fourier mass outside this ball is therefore $\eta_g$ in~\eqref{eq:eta}. Substitution into Theorem~\ref{Theorem:Main} gives the success and running-time bounds.

        For $m\ge g$, the lower-tail Chernoff bound gives
        \[
        \eta_g\le\exp\left(-\frac{(m-g)^2}{2m}\right).
        \]
        The value $m_*$ solves $(m_*-g)^2=2Lm_*$. Hence $m\ge m_*$ implies $\eta_g\le P_{\mathrm{Dec}}/64$, and
        \[
        P_{\mathrm{Dec}}(1-\eta_g)-2\sqrt{\eta_gP_{\mathrm{Dec}}(1-P_{\mathrm{Dec}})}
        \ge\frac{63P_{\mathrm{Dec}}}{64}-\frac{P_{\mathrm{Dec}}}{4}
        =\frac{47P_{\mathrm{Dec}}}{64}.
        \]
        Finally, $\tau\mapsto\tau^\perp$ is decreasing and involutive on $[0,1-1/q]$, so~\eqref{eq:taustar} inverts $\tau^\perp=1-m_*/q$.
    \end{proof}

    \subsection{Instantiations with concrete decoders}

\subsubsection{Classical decoders}
	
	We use the Berlekamp--Welch and Guruswami--Sudan decoders. Measuring a QDP input with product amplitude $u^{\otimes q}$ produces a classical error whose weight is distributed as $\operatorname{Bin}(q,\tau)$.

\begin{proposition}[Classical decoders]\label{Proposition:Decoders}
    The following decoders can be used as $\QDP$ solvers after measuring the input state in the computational basis:
    \begin{itemize}\setlength\itemsep{-0.2em}
        \item The Berlekamp--Welch algorithm corrects up to $\lfloor(q-k)/2\rfloor$ errors in $\RS[q,k]_q$.
        \item The Guruswami--Sudan algorithm list decodes below the Johnson radius $q-\sqrt{q(k-1)}$~\cite{guruswami1998improved}. Choosing uniformly from a polynomial-size list containing the hidden codeword gives inverse-polynomial success.
    \end{itemize}
\end{proposition}

In each case, QDP success includes the probability that the measured error lies within the decoder's radius. For the parameters below, this probability and the list-selection probability are inverse polynomial in $q$. We use $k=O(q^{1/4})$ and $k/\log q\to\infty$, as in the asymptotic Cheng--Wan regime.

\begin{proposition}[Using the Berlekamp--Welch decoder]\label{prop:BW}
Set $\tau=k/(2q)$ and $\tau_{\rm BW}=\tau^\perp$. Combining the decoder with Theorem~\ref{thm:RS-instantiation} gives inverse-polynomial success for $\IBDD(\RS[q,k]_q,q-g)$ whenever the margin condition~\eqref{eq:margin} holds. Its mean number of agreements is
\[
 m_{\rm BW}=q(1-\tau_{\rm BW})=k/2+2\sqrt{k/2}+1+o(1).
\]
\end{proposition}
\begin{proof}
The dual code is $\RS[q,q-k]_q$, for which Berlekamp--Welch corrects up to $\lfloor k/2\rfloor$ errors. At $\tau=k/(2q)$, the probability of this event is bounded below by a constant. Apply Theorem~\ref{thm:RS-instantiation} and evaluate \eqref{eq:p0}.
\end{proof}

\begin{proposition}[Using the Guruswami--Sudan decoder]\label{prop:GS}
Set $\tau=1-\sqrt{(q-k)/q}$ and $\tau_{\rm GS}=\tau^\perp$. The corresponding quantum reduction has inverse-polynomial success at the radii satisfying~\eqref{eq:margin}. Its mean number of agreements is
\[
 m_{\rm GS}=q(1-\tau_{\rm GS})=k/2+2\sqrt{k/2}+1+o(1).
\]
\end{proposition}
\begin{proof}
Apply Guruswami--Sudan to $\RS[q,q-k]_q$ after measuring the QDP input. At the stated parameter, the probability of lying within a polynomial-time list-decoding radius is inverse polynomial. Selecting uniformly from the polynomial-size list preserves inverse-polynomial success. Theorem~\ref{thm:RS-instantiation} and \eqref{eq:p0} then give the claim; the expansion uses $k=O(q^{1/4})$.
\end{proof}

\subsubsection{Unambiguous state discrimination}

We use the coordinatewise unambiguous state discrimination decoder of~\cite[Proposition~16]{CT24}. See~\cite{Ivanovic1987} for the original USD measurement.

\begin{proposition}[Unambiguous state discrimination]\label{prop:USD}
Set $\tau_{\rm USD}=(q-1)(q-k)/q^2$ and choose $\tau$ with $\tau^\perp=\tau_{\rm USD}$. The resulting QDP solver and Theorem~\ref{thm:RS-instantiation} give inverse-polynomial success at the radii satisfying~\eqref{eq:margin}. The mean number of agreements is
\[
 m_{\rm USD}=q(1-\tau_{\rm USD})=k+1-k/q.
\]
\end{proposition}
\begin{proof}
For the one-coordinate states, unambiguous discrimination returns the correct symbol with probability $q\tau^\perp/(q-1)$ and otherwise returns an erasure. At the stated parameter, the mean number of recovered coordinates is $q-k$. With probability bounded below by a constant in this regime, at least $q-k$ coordinates are recovered, which suffices to interpolate the hidden codeword of $\RS[q,q-k]_q$. The binomial-tail reduction gives the claim.
\end{proof}

\subsubsection{Performance for the Cheng--Wan parameters}

    \paragraph{Parameter recap.}
    The Cheng--Wan reduction uses $\RS[q,k]_q$ with $k=3h+4$ and target radius $q-4h-4$. Writing $g=4h+4$ and $\widetilde h=h/q$, the required mean agreement count satisfies
    \[
    \tau^\perp\le1-g/q\quad\Longleftrightarrow\quad\tau\ge(1-g/q)^\perp.
    \]
    To obtain the success guarantee in Theorem~\ref{thm:RS-instantiation}, the mean must exceed $g$ by the margin in~\eqref{eq:margin}.

    \paragraph{The approximation $\tau^\perp\approx1-\tau$ for small $\tau$.}
    Write $\tau=t/q$ with $t/q\to0$. The exact Fourier formula gives
    \[
    \tau^\perp=1-\frac{t+2\sqrt t+1}{q}
    +O\!\left(\frac{t^{3/2}+t+\sqrt t}{q^2}\right).
    \]
    If $t\to\infty$, the leading agreement count is $t$. At $k=3h+4$, Berlekamp--Welch and Guruswami--Sudan give $3h/2+o(h)$ mean agreements, and USD gives $3h+o(h)$. The Cheng--Wan target requires at least $4h+4$ agreements, together with a tail margin. Table~\ref{Table:1} summarizes these thresholds.

    \paragraph{Condition on $h$.}
    The Cheng--Wan reduction requires $2\le h\le q^{1/4}-2$. For an inverse-polynomial QDP success probability and $h/\log q\to\infty$, equation~\eqref{eq:margin} gives
    \[
    m_*=4h+O\!\left(\sqrt{h\log q}+\log q+1\right).
    \]
    Equation~\eqref{eq:taustar} then yields $q\tau_*=(4+o(1))h$. A sufficient choice is $\tau=(4+\varepsilon)h/q$ for fixed $\varepsilon>0$, giving $\eta_g=\exp(-\Omega_\varepsilon(h))$. One may also take $\varepsilon\to0$ provided $\varepsilon^2h/\log q\to\infty$. When $h$ is smaller, the exact binomial-tail bound~\eqref{eq:eta} must be used.

    \begin{table}[h!]
        \centering
        \renewcommand{\arraystretch}{1.4}
        \small
        \begin{tabular}{|l|c|c|l|}
            \hline
            Decoder & QDP noise $\tau$ & IBDD error fraction $\tau'$ & Remark \\
            \hline
            Berlekamp--Welch & $\approx\tfrac{3\th}{2}$ & $\approx1-\tfrac{3\th}{2}$ & \\
            Guruswami--Sudan & $\approx\tfrac{3\th}{2}$ & $\approx1-\tfrac{3\th}{2}$ & \\
            USD & $\approx3\th$ & $\approx1-3\th$ & \\
            \hline
            Required (Cheng--Wan) & $\approx4\th$ & $\approx1-4\th$ & DLOG, under Conj.~\ref{ass:heuristic} \\
            \hline
        \end{tabular}
        \caption{QDP noise on $\RS[q,q-k]_q$ and resulting IBDD error fraction on $\RS[q,k]_q$, with $k=3h+4$, $\widetilde h=h/q$ and $h/\log q\to\infty$. The entries give the leading terms; the success guarantee includes the tail margin in~\eqref{eq:margin}. The exact Cheng--Wan radius is $q-4h-4$. The three decoding guarantees fall short of this target by a constant factor.}
        \label{Table:1}
    \end{table}
	
	\paragraph{From uniform to Cheng--Wan instances.}
	Theorem~\ref{thm:RS-instantiation} gives a success guarantee for uniformly random received words, whereas the Cheng--Wan reduction produces the structured words $\yv_0^{(i)}$.
	
	Two settings illustrate how a guarantee for fixed instances can be obtained. The first is when the decoder is exact, in which case the reduction handles arbitrary $\BDD$ instances. The second is when the error function is itself randomized in a way that decouples it from the instance, as in the random-oracle construction of~\cite{YZ24}. Outside these cases, the reduction can fail. Explicit counterexamples appear in~\cite{CT24,BCT25}: one can construct approximate decoders that, applied to fixed received words, are correlated with the instance in such a way that the dual-code information is destroyed after the Fourier transform. Passing to $\IBDD$ avoids this failure mode by randomizing the starting point, so the decoder cannot be tuned to a specific coset.
	
	The Cheng--Wan instances are sampled from a structured family. The starting point $\yv_0^{(i)}$ depends on $i$, which ranges over $\{0,\ldots,q^h-2\}$, and the discrete logarithm itself can be randomized via the self-reduction of Theorem~\ref{thm:random-self-reduction-general}. This randomization preserves the Cheng--Wan distribution and does not produce uniformly random received words.
	
\begin{conjecture}[Transfer to the Cheng--Wan distribution] \label{ass:heuristic} Fix an efficient QDP decoder \(Dec\), and let \(\mathcal A_{Dec}\) be the BDD algorithm obtained by applying Regev's reduction with \(Dec\) at the Cheng--Wan parameters \(g=4h+4\). We conjecture that if \[ \Pr_{\yv_0\leftarrow\mathbb F_q^q} \left[ \mathcal A_{Dec}(\yv_0) \text{ outputs a codeword within distance }q-g \text{ of }\yv_0 \right] \geq \frac{1}{\poly(q)}, \] then \[ \Pr_{\yv_0\leftarrow\mathcal D_{q,h,g}} \left[ \mathcal A_{Dec}(\yv_0) \text{ outputs a codeword within distance }q-g \text{ of }\yv_0 \right] \geq \frac{1}{\poly(q)}. \] Here, \(\mathcal D_{q,h,g}\) is the Cheng--Wan received-word distribution from \Cref{def:CW-distribution}. Both probabilities also include the internal randomness and measurement outcomes of \(\mathcal A_{Dec}\). \end{conjecture}

Informally, the conjecture asserts that the BDD algorithm induced by a QDP decoder retains non-negligible success probability when uniformly random received words are replaced by structured Cheng--Wan received words. The Cheng--Wan instances are structured, since they arise from the algebraic form of the Cheng--Wan reduction rather than from a uniform BDD distribution. However, this structure seems unrelated to the way our algorithm uses the instances. Conjecture~\ref{ass:heuristic} is therefore not a claim that the Cheng--Wan distribution is statistically close to the uniform BDD distribution; it concerns only the success probability of the induced BDD algorithm.

\begin{theorem}[DLOG via QDP] \label{thm:DLOG-via-Regev-heuristic} Assume the field representation and primitive generator of \Cref{thm:abelian-DLOG-BDD}, with \(2\leq h\leq q^{1/4}-2\). Suppose there exists an efficient QDP decoder for \(\RS[q,q-(3h+4)]_q\) whose induced BDD algorithm satisfies the transfer property in \Cref{ass:heuristic}. If the decoder succeeds with probability \(P_{\mathrm{Dec}}\geq 1/\poly(q)\) at a noise parameter satisfying the margin condition~\eqref{eq:margin} for \(g=4h+4\), then discrete logarithms in \(\F_{q^h}^{\times}\) can be computed with high probability in expected time polynomial in \(q\). If \(h/\log q\to\infty\), a sufficient leading-order noise threshold is \(\tau=(4+o(1))h/q\), subject to the margin condition above. \end{theorem}
\begin{proof} Theorem~\ref{thm:RS-instantiation} gives inverse-polynomial success probability for the induced BDD algorithm on uniformly random received words. By Conjecture~\ref{ass:heuristic}, this algorithm retains inverse-polynomial success probability on the Cheng--Wan distribution. Substituting this success guarantee into \Cref{thm:abelian-DLOG-BDD} gives expected running time polynomial in \(q\) under the stated assumptions. \end{proof}

\section{The Pretty Good Measurement}
	\label{sec:PGM}
	
	For every promised $\BDD$ instance, we construct an efficiently preparable error state such that coherent access to its PGM solves the instance through Regev's reduction. Both the state and the PGM depend on the instance. The result applies to the Cheng--Wan instances and the promised YES instances in the $\NP$-hard families of Theorem~\ref{thm:ggg18}; no efficient general implementation of the PGM is known.

    \subsection{The Pretty Good Measurement for the QDP problem}
	
	\begin{definition}
		Let $\{\ket{\psi_\cv}\}_{\cv \in \C}$ be a set of quantum pure states. For the uniform prior, the Pretty Good Measurement associated with these states is the POVM $\{M_\cv\}_{\cv \in \C}$ where
		$$M_\cv = \rho^{-\frac{1}{2}} \kb{\psi_\cv} \rho^{-\frac{1}{2}} \quad \textrm{with} \quad \rho = \sum_{\cv \in \C} \kb{\psi_{\cv}}.$$
        The inverse is taken on the support of $\rho$, and the POVM is completed on its orthogonal complement.
	\end{definition}
	
	\begin{proposition}[\cite{BCT25}]\label{Proposition:PGM}
		Let $\Gm \in \F_q^{k \times n}$ be a generator matrix of a code $\C$ and let $f : \F_q^n \rightarrow \mathbb{C}$ with $\norm{f}_2 = 1$. Let $\{\ket{\psi_{\sv}}\}_{\sv \in \F_q^k}$ be the set of quantum pure states where $\ket{\psi_\sv} = \sum_{\ev \in \F_q^n} f(\ev) \ket{\sv\Gm + \ev}$. For each $\uv \in \F_q^k$, let
		$$\ket{W_{\uv}} = \sum_{\yv \in \C^\perp_{\uv}} \hf(\yv) \ket{\yv} \quad \text{and} \quad \ket{\widetilde{W}_{\uv}} = \frac{\ket{W_{\uv}}}{\norm{\ket{W_{\uv}}}}.$$
		If $\norm{\ket{W_\uv}} \neq 0$ for each $\uv \in \F_q^k$, then the PGM in the Fourier basis associated with $\{\ket{\psi_{\sv}}\}_{\sv \in \F_q^k}$ is the measurement $\{\kb{\widehat{Y}_\sv}\}_{\sv}$ with
		$\ket{\widehat{Y}_\sv} = \frac{1}{\sqrt{q^k}}\sum_{\uv \in \F_q^k} \chi_{\uv}(\sv) \ket{\widetilde{W}_{\uv}}$.
	\end{proposition}
	
	\subsection{Solving promised BDD instances with the Pretty Good Measurement}
	
	\begin{theorem}\label{thm:PGM-solves-BDD}
		For any promised $\BDD(\C,t,\yv_0)$ instance, there is an efficiently preparable, instance-dependent error state for $\QDP(\C^\perp,f)$ such that coherent implementation of its PGM and its inverse solves the BDD instance through Regev's reduction, with constant expected repetition cost.
	\end{theorem}
	
	\begin{proof}
		Let $\mathcal D=\C^\perp$ and let $\Gm\in\F_q^{k\times n}$ be a full-rank generator matrix of $\mathcal D$, so $\Gm$ is a parity-check matrix of the target code $\C$. If $k=0$, then $\C=\F_q^n$ and the problem is trivial. Assume $k\ge1$ and set $\uv_0=\Gm\yv_0^\top$. We construct a suitable error function $f$ and verify the hypotheses of Proposition~\ref{Proposition:PGM}. Define
		$$g(\xv) = \left\{
		\begin{array}{cl}
			1 & \text{if } \Gm {\xv}^\top \neq \uv_0 \text{ or } \norm{\xv} \le t, \\
			0 & \text{otherwise,}
		\end{array}\right.
		$$
		and set $\hf = \frac{g}{\norm{g}_2}$.
		
		\begin{lemma}\label{Lemma:StateConstruction}
			One can efficiently construct the state $\sum_{\ev \in \F_q^n} f(\ev)\ket{\ev}$.
		\end{lemma}
		\begin{proof}
			Since $g$ is efficiently computable, the state $\ket{A} = \frac{1}{\sqrt{q^n}}\sum_{\xv \in \F_q^n} \ket{\xv}\ket{g(\xv)}$ is easily prepared. Let $T = \{\xv : g(\xv) = 1\}$. We have $|\{\xv : \Gm\xv^{\top} \neq \uv_0\}| = q^n - q^{n-k}$, so $|T| > q^n - q^{n-k}$ (using the assumption that a BDD solution exists).
			
			We prepare $\ket{A}$ and measure the second register. If we measure $0$, we restart. Otherwise we obtain the state
			$$\frac{1}{\sqrt{|T|}} \sum_{\xv \in T} \ket{\xv} = \sum_{\xv \in \F_q^n} \hf(\xv)\ket{\xv}.$$
			Applying an inverse quantum Fourier transform yields $\sum_{\ev \in \F_q^n} f(\ev)\ket{\ev}$. The success probability at each attempt is $|T|/q^n \ge 1-q^{-k}\ge1/2$, so the procedure is efficient.
		\end{proof}
		
		We now verify the non-vanishing condition on the dual coset states $\ket{W_\uv} = \sum_{\yv \in \mathcal D^\perp_\uv} \hf(\yv)\ket{\yv}$ required by Proposition~\ref{Proposition:PGM}:
		\begin{itemize}
			\item If $\uv \neq \uv_0$, then $g(\xv) = 1$ for all $\xv \in \mathcal D^\perp_{\uv}$ (since $\Gm\xv^\top \neq \uv_0$), so $\hf(\xv) = \frac{1}{\sqrt{|T|}}$ there, and $\norm{\ket{W_{\uv}}} = \sqrt{\frac{q^{n-k}}{|T|}} > 0$.
			\item If $\uv = \uv_0$, let $\zv$ be an error vector corresponding to a BDD solution (so $\Gm\zv^\top = \uv_0$ and $\norm{\zv} \le t$). Then $g(\zv) = 1$ by definition, so $\norm{\ket{W_{\uv_0}}} > 0$.
		\end{itemize}
		We can therefore apply Proposition~\ref{Proposition:PGM}. The PGM vectors $\ket{Y_\sv}=\mathrm{QFT}^{-1}\ket{\widehat Y_\sv}$ are orthonormal and span the states $\ket{\psi_\sv}$. A coherent implementation of the measurement gives a unitary
		$$U_{PGM} : \ket{Y_\sv}\ket{\yv} \rightarrow \ket{Y_\sv}\ket{\yv + \sv}.$$
		The algorithm proceeds as follows.
		\begin{enumerate}
			\item \emph{Prepare the phase-encoded superposition.} Construct the state
			$$\ket{\Omega_1} = \frac{1}{\sqrt{q^k}} \sum_{\sv \in \F_q^k} \overline{\chi_{\sv}(\uv_0)} \ket{\psi_{\sv}}\ket{\sv}, \quad \text{with } \ket{\psi_\sv} = \sum_{\ev \in \F_q^n} f(\ev)\ket{\sv\Gm + \ev}.$$
			
			\item \emph{Expand in the PGM basis.} Write $\ket{\psi_{\sv}} = \sum_{\sv'} \gamma_{\sv,\sv'} \ket{Y_{\sv'}}$, so that
			$$\ket{\Omega_1} = \frac{1}{\sqrt{q^k}} \sum_{\sv \in \F_q^k} \overline{\chi_{\sv}(\uv_0)} \left(\sum_{\sv' \in \F_q^k} \gamma_{\sv,\sv'}\ket{Y_{\sv'}}\right)\ket{\sv}.$$
			
			\item \emph{Apply $U_{PGM}^\dagger$ to subtract the decoded label.} Obtain
			$$\ket{\Omega_2} = \frac{1}{\sqrt{q^k}} \sum_{\sv, \sv' \in \F_q^k} \overline{\chi_{\sv+\sv'}(\uv_0)}\gamma_{\sv + \sv',\sv'}\ket{Y_{\sv'}}\ket{\sv}.$$
			
			\item \emph{Measure and post-select on $0^k$.} The probability that the outcome is $0^k$ equals $\frac{1}{q^k}\sum_{\sv \in \F_q^k} |\gamma_{\sv,\sv}|^2$. By Lemma~\ref{Lemma:GammaEquality} (proved below), $\gamma_{\sv,\sv} = \Gamma \ge 1 - \frac{1}{q^k}$ for all $\sv$, where $\Gamma$ is real and positive. Hence postselection succeeds with probability $\Gamma^2\ge(1-q^{-k})^2\ge1/4$, giving constant expected repetition cost. The post-measurement state is
			$$\ket{\Omega_3} = \frac{1}{\sqrt{q^k}}\sum_{\sv' \in \F_q^k} \overline{\chi_{\sv'}(\uv_0)} \ket{Y_{\sv'}}.$$
			If the outcome is not $0^k$, restart from step~1.
			
			\item \emph{Apply the quantum Fourier transform to recover the target coset.} Obtain
			$$\ket{\Omega_4} = \frac{1}{q^k} \sum_{\sv', \uv \in \F_q^k} \chi_{\sv'}(\uv-\uv_0)\ket{\widetilde{W}_{\uv}} = \ket{\widetilde{W}_{\uv_0}}.$$
			
			\item \emph{Measure to obtain a solution.} Let $S = \{\xv \in \mathcal D^\perp_{\uv_0} : \norm{\xv} \le t\}$. We have
			$$\ket{\widetilde{W}_{\uv_0}} = \frac{1}{\norm{\ket{W_{\uv_0}}}} \sum_{\xv \in \mathcal D^\perp_{\uv_0}} \hf(\xv)\ket{\xv} = \frac{1}{\sqrt{|S|}} \sum_{\xv \in S} \ket{\xv}.$$
			Measuring in the computational basis yields a uniformly random error $\xv\in S$; the codeword $\yv_0-\xv$ solves the BDD instance.
		\end{enumerate}
		It remains to prove the lemma used in step~4.
		\begin{lemma}\label{Lemma:GammaEquality}
			There exists a real $\Gamma \ge 1 - \frac{1}{q^k}$, independent of $\sv$, such that $\gamma_{\sv,\sv} = \Gamma$ for all $\sv \in \F_q^k$.
		\end{lemma}
		\begin{proof}
			Write $w_\uv=\|\ket{W_\uv}\|$. We compute
			\begin{align*}
				\gamma_{\sv,\sv} = \langle Y_{\sv}\mid\psi_{\sv}\rangle
				= \langle\widehat Y_{\sv}\mid\widehat\psi_{\sv}\rangle
				= \frac{1}{\sqrt{q^k}}\sum_{\uv \in \F_q^k} \langle\widetilde W_\uv\mid W_\uv\rangle
				= \frac{1}{\sqrt{q^k}}\sum_{\uv \in \F_q^k} w_{\uv} =: \Gamma,
			\end{align*}
			which is independent of $\sv$. For $\uv \neq \uv_0$, since $|T| \le q^n$ we have $w_{\uv} = \sqrt{q^{n-k}/|T|} \ge q^{-k/2}$, and therefore
			\begin{equation*}
				\Gamma = \frac{1}{\sqrt{q^k}}\sum_{\uv \in \F_q^k} w_{\uv} \ge \frac{q^k - 1}{q^k} = 1 - \frac{1}{q^k}. \qedhere
			\end{equation*}
		\end{proof}
	\end{proof}
	
	Theorem~\ref{thm:PGM-solves-BDD} applies to the Cheng--Wan instances and to the promised YES instances in the $\NP$-hard families of Theorem~\ref{thm:ggg18}. It requires coherent access to the PGM and its inverse, rather than access only to a destructive measurement. The error amplitude is chosen for the target instance, so this result does not establish Conjecture~\ref{ass:heuristic} for the Bernoulli-QDP reduction.

    \section*{Acknowledgments.}
	
	MIFG was supported by the National Science Foundation under NSF CAREER Award CCF-2048204. MIFG thanks Urmila Mahadev, John Preskill, Itay Shalit, Stephen Jordan and Elena Grigorescu for useful discussions. AC thanks Daniel Augot and François Morin for helpful discussions, and acknowledges funding from the French PEPR integrated projects EPIQ (ANR-22-PETQ-007), PQTLS (ANR-22-PETQ-008) and HQI (ANR-22-PNCQ-0002) all part of the France 2030 plan. The authors thank Noah Shutty for pointing out the applicability of the Kleinjung--Wesolowski algorithm to the Cheng--Wan parameter regime.

	\bibliographystyle{alpha}
	\bibliography{Bibv2}
	
	\appendix
	
	\section{Relation collection in the Cheng--Wan reduction}

The following module version of the ordinary-and-shifted sampling argument is used in Corollary~\ref{cor:cheng-wan-augmented-relations}; compare~\cite[Lemma~4.1 and the following remark]{Pomerance1987}.

\begin{lemma}[Generation by ordinary and shifted samples]
\label{lem:pomerance-relations}
Let \(N\ge2\), \(m\ge1\), \(R=\Z/N\Z\), and let \(\nu\) be any
distribution on \(R^m\). Fix \(0<\delta<1\), and set
\[
 \ell_N=\lceil\log_2N\rceil,\qquad
 A=\left\lceil16\left(m+\ln\frac{2\ell_N}{\delta}\right)\right\rceil,
 \qquad
 B=\left\lceil\log_2\frac{2m\ell_N}{\delta}\right\rceil.
\]
Draw \(A\) independent ordinary samples \(v_i\sim\nu\).
For each \(j\in\{1,\ldots,m\}\), independently draw \(B\) further
samples \(w_{j,b}\sim\nu\), and include the vectors \(e_j+w_{j,b}\).
With probability at least \(1-\delta\), all collected vectors generate
the entire \(R\)-module \(R^m\).
No spanning assumption on \(\supp(\nu)\) is required.
\end{lemma}

\begin{proof}
Fix a prime \(\ell\mid N\) and reduce all vectors modulo \(\ell\).
Let \(U_i\) be the span of the first \(i\) ordinary samples over
\(\F_\ell\). As long as \(\nu(U_i)<1/2\), the next sample increases
the dimension with probability greater than \(1/2\).
Couple these increments with independent Bernoulli-\(1/2\) variables
until the mass reaches \(1/2\); after that, it can never decrease.
Since the dimension is at most \(m\), this gives
\[
 \Pr[\nu(U_A)<1/2]
 \le \Pr[\operatorname{Bin}(A,1/2)<m]
 \le e^{-A/16}
 \le \frac{\delta}{2\ell_N}.
\]
The middle inequality follows from the multiplicative Chernoff bound,
using \(A\ge4m\).

Condition on ordinary samples with \(\nu(U_A)\ge1/2\).
For any coordinate \(j\), the probability that none of its \(B\)
fresh samples lies in \(U_A\) is at most \(2^{-B}\).
Except with probability at most \(m2^{-B}\le\delta/(2\ell_N)\),
each coordinate has such a sample. For that sample,
\[
 e_j=(e_j+w_{j,b})-w_{j,b}
\]
lies in the span of the collected rows. Hence those rows span
\(\F_\ell^m\).

There are at most \(\ell_N\) distinct prime divisors of \(N\).
A union bound gives full span modulo every such prime with probability
at least \(1-\delta\).
Full rank modulo \(\ell\) supplies a minor with determinant coprime to
\(\ell\), so the rows generate modulo every power of \(\ell\)
dividing \(N\). The Chinese remainder theorem then gives generation
over \(R\).
The primes are used only in this proof; the sampling algorithm does
not need to factor \(N\).
\end{proof}

For computation, lift the relation rows to integers and adjoin
the rows of \(NI_m\). They generate the integer lattice \(\Z^m\)
exactly when the relation rows generate \(R^m\).
Hermite or Smith normal form with transformation data tests this and
expresses each \(e_j\) as a combination of relation rows modulo \(N\),
in polynomial bit complexity. Ordinary Gaussian elimination over a
field is not an adequate substitute.

\section{NP-hardness proof of BDD for Reed--Solomon codes}\label{appendix:Subset-sum}

    We use the following finite-field reductions in the proof of Theorem~\ref{thm:ggg18}.

    \begin{lemma}[1-in-3SAT to MSS$(d)$ over a finite field {\cite[Section~7]{GandikotaGhaziGrigorescu2018}}]
        \label{lem:13sat_to_MSS_Q}
        A 1-in-3SAT instance with $v$ variables can be reduced to an MSS$(d)$ instance $(A,K,m_1,\ldots,m_d)$ over an explicitly represented finite field $F$ of characteristic $p>d$, with
        \[
        |A|=N=v(2^{d+1}-2),\qquad K=N/2.
        \]
        The elements of $A$ are distinct, and the MSS instance is YES if and only if the 1-in-3SAT instance is satisfiable. The construction takes time polynomial in $v$ and $d!$, with field encodings of polynomial length in these parameters. The characteristic can be chosen as $p=O(d!)$, with a constant adjustment for the smallest values of $d$.
    \end{lemma}

    This is the finite-field construction of~\cite[Section~7]{GandikotaGhaziGrigorescu2018}, which carries out the reduction directly in a polynomial basis over $\F_p$.

    \begin{lemma}[Translation to nonzero evaluation elements]
        \label{lem:Q_to_Z}
        Consider the padded instance $A'\subset E$ of Lemma~\ref{lem:mss_padding_fixed_k}. Let $\beta=1+\theta$, set $A''=\{a+\beta:a\in A'\}$, and define
        \[
        m_0=K\cdot1_E,\qquad
        m_r''=\sum_{j=0}^r\binom rj\beta^{r-j}m_j\quad(1\le r\le d).
        \]
        The translated MSS instance is equivalent to the padded instance, and its evaluation elements are distinct and nonzero.
    \end{lemma}
    \begin{proof}
        For subsets of size $K$, the binomial theorem proves equivalence of the moment equations; translation by $-\beta$ gives the inverse transformation. An original element $a+\beta$ has a nonzero $\theta$ coefficient. Each translated dummy $\theta^j+1+\theta$ has constant coefficient $1$, including when $j=1$ in characteristic two. All elements are therefore nonzero, and translation preserves distinctness.
    \end{proof}

    \begin{lemma}[MSS$(d)$ over a finite field to Reed--Solomon BDD]
        \label{lem:Z_to_Fq}
        Let $(A'',K,m_1'',\ldots,m_d'')$ be an MSS$(d)$ instance over a finite field $E$ of characteristic $p>d$, where the $N'$ elements of $A''$ are distinct and nonzero and $1\le d\le K\le N'$. There is an equivalent decision BDD instance for a Reed--Solomon code over $E$, with
        \begin{equation}
        n=N'+1,\qquad k=K-d+1,\qquad t=n-k-d=N'-K.
        \label{eq:npdimensions}
        \end{equation}
        The reduction takes time polynomial in the input size.
    \end{lemma}
    \begin{proof}
        Since $p>d$, Newton's identities reversibly convert the moment targets into elementary-symmetric targets $E_1,\ldots,E_d$:
        \[
        E_0=1,\qquad jE_j=\sum_{i=1}^j(-1)^{i-1}E_{j-i}m_i''\quad(1\le j\le d).
        \]
        Following~\cite[Lemma~2.4]{GandikotaGhaziGrigorescu2018}, set
        \[
        F_d(X)=\sum_{j=0}^{d-1}(-1)^jE_jX^{d-j}.
        \]
        At evaluation point $a^{-1}$, for $a\in A''$, assign the received value $-F_d(a)$; at evaluation point $0$, assign $(-1)^dE_d$.

        A size-$K$ solution has root polynomial $G(X)=\prod_{a\in S}(X-a)$ with the prescribed first $d$ coefficients. The polynomial
        \[
        p(X)=\frac{X^K G(1/X)-X^dF_d(1/X)}{X^d}
        \]
        has degree at most $K-d=k-1$ and agrees with the received word at the selected $K$ inverse points and at zero.

        Conversely, a polynomial $p$ of degree at most $K-d$ that agrees in at least $K+1$ positions must agree at zero. Otherwise the monic degree-$K$ polynomial
        \[
        G(X)=X^{K-d}\bigl(p(1/X)+F_d(X)\bigr)
        \]
        would have at least $K+1$ distinct roots. Agreement at zero fixes the coefficient corresponding to $E_d$, and the $K$ inverse-point agreements give a size-$K$ set of roots with the required symmetric sums. Newton's identities recover the required moments. This proves both directions.
    \end{proof}

\end{document}